\newif\iffinal
\finaltrue

\iffinal
  \documentclass[a4paper, 11pt, final]{article}
\else
  \documentclass[a4paper, 11pt]{article}
\fi

\usepackage{etex}
\usepackage[utf8]{inputenc}
\usepackage[T1]{fontenc}

\usepackage[english]{babel}

\usepackage[a4paper,top=2cm,bottom=2cm,left=15mm,right=2cm,marginparwidth=3.5cm]{geometry}

\usepackage[charter]{mathdesign}
\usepackage{eucal}
\usepackage{microtype}
\usepackage{nicefrac}

\usepackage{booktabs}

\usepackage[shortlabels, inline]{enumitem}
\setlist{itemsep=-4pt,topsep=0pt}

\usepackage{xcolor}
\colorlet{darkgreen}{green!50!black}
\colorlet{dg}{darkgreen}
\colorlet{medgray}{gray!75}
\colorlet{lightgray}{gray!30}
\definecolor{linkcol}{rgb}{0,0,0.5} 
\definecolor{citecol}{rgb}{0.5,0,0}

\usepackage[bookmarks=false, pdftex, colorlinks=true, allcolors=blue]{hyperref}
\hypersetup{
pdftex,
bookmarksopen=false,
pdftoolbar=false, 
pdfmenubar=true, 
pdfhighlight=/O, 
colorlinks=true, 
pdfpagemode=UseNone, 
pdfpagelayout=SinglePage, 
pdffitwindow=true, 
linkcolor=linkcol, 
citecolor=citecol, 
urlcolor=linkcol 
}

\usepackage[useprefix, maxcitenames=2, mincitenames=1, maxbibnames=9]{biblatex}
\let\citet\textcite
\let\citep\parencite

\usepackage{amsfonts,amsmath,amsthm}
\usepackage{graphicx}
\usepackage[vlined,ruled,linesnumbered]{algorithm2e}
\usepackage[capitalize, nameinlink]{cleveref}

\iffinal
  \usepackage[disable]{todonotes}
\else
  \usepackage{todonotes}
  \usepackage[notref,notcite]{showkeys}
\fi

\newtheorem{theorem}{Theorem}
\newtheorem{definition}{Definition}
\newtheorem{lemma}{Lemma}
\newtheorem{proposition}{Proposition}

\newtheorem{corollary}{Corollary}
\newtheorem{observation}{Observation}
\newtheorem{construction}{Construction}

\newcommand{\emb}{\phi}
\newcommand{\belo}[1]{\ensuremath{<_{#1}}}
\newcommand{\beloeq}[1]{\ensuremath{\leq_{#1}}}
\newcommand{\abov}[1]{\ensuremath{>_{#1}}}
\newcommand{\aboveq}[1]{\ensuremath{\geq_{#1}}}

\usepackage{xspace}

\newcommand{\TC}{\textsc{Tree Containment}\xspace}

\newcommand{\Nh}{N_{\textsc{host}}}
\newcommand{\Tg}{T_{\textsc{guest}}}
\newcommand{\Fg}{T_{\textsc{guest}}} 

\let\emptyset\varnothing
\newcommand{\outA}[2][]{\ensuremath{A^+\ifx\relax#1\relax\else_{#1}\fi(#2)}}
\newcommand{\inA}[2][]{\ensuremath{A^-\ifx\relax#1\relax\else_{#1}\fi(#2)}}
\newcommand{\outdeg}[2][]{\ensuremath{\operatorname{deg}^+\ifx\relax#1\relax\else_{#1}\fi(#2)}}
\newcommand{\indeg}[2][]{\ensuremath{\operatorname{deg}^-\ifx\relax#1\relax\else_{#1}\fi(#2)}}

\newcommand{\replace}[3]{\ensuremath{#1[\mbox{\ensuremath{#2 \to #3}}]}}

\usepackage{xcolor}
\colorlet{darkgreen}{green!50!black}
\colorlet{dg}{darkgreen}
\colorlet{medgray}{gray!75}
\colorlet{lightgray}{gray!30}
\colorlet{llightgray}{lightgray!30}
\colorlet{bagcol}{gray!70}
\colorlet{pastcol}{gray!20}
\definecolor{linkcol}{rgb}{0,0,0.4} 
\definecolor{citecol}{rgb}{0.5,0,0} 

\usepackage{tikz}
\usetikzlibrary{arrows.meta,shapes,positioning,calc,decorations.pathreplacing,decorations.markings,decorations.pathmorphing,patterns}
\pgfdeclarelayer{background2}
\pgfdeclarelayer{background}
\pgfdeclarelayer{foreground}
\pgfsetlayers{background2,background,main,foreground}

\tikzstyle{hidden}=[opacity=0]
\tikzstyle{fade}=[opacity=0.2]
\tikzstyle{nonsol}=[dashed]

\tikzstyle{bold}=[draw, line width=2pt]
\tikzstyle{gamma}=[draw, line width=5pt]
\tikzstyle{optional}=[dashed]
\tikzstyle{path}=[decorate, decoration={snake, amplitude=.6mm}]

\tikzstyle{small}=[inner sep=2pt]
\tikzstyle{tiny}=[inner sep=1.7pt]
\tikzstyle{textnode}=[inner sep=0pt]

\tikzstyle{triangle}=[draw, regular polygon, regular polygon sides=3]
\tikzstyle{hex}=[draw, regular polygon, regular polygon sides=6]
\tikzstyle{vertex}=[circle, draw, fill=white]
\tikzstyle{reti}=[vertex, fill=black]
\tikzstyle{leaf}=[vertex, rectangle]
\tikzstyle{leaf2}=[vertex, triangle]

\tikzstyle{smallvertex}=[vertex, small]
\tikzstyle{smallvertex0}=[vertex, star, small]
\tikzstyle{smallvertex1}=[vertex, small, diamond]
\tikzstyle{smallvertex2}=[vertex, triangle, small]
\tikzstyle{smallvertex3}=[vertex, hex, inner sep=3pt]
\tikzstyle{smallvertex4}=[smallvertex]
\tikzstyle{smallleaf}=[leaf, inner sep=3.3pt]
\tikzstyle{smallleaf2}=[leaf2, inner sep=1.7pt]
\tikzstyle{smalltriangle}=[triangle, inner sep=1.5pt]
\tikzstyle{smallreti}=[reti, small]

\tikzstyle{tinyvertex}=[vertex, tiny]

\tikzstyle{normal}=[smallvertex, fill=black]

\tikzstyle{edge}=[draw,-]
\tikzstyle{matching}=[edge,line width=3pt]
\tikzstyle{arc}=[draw,arrows={-Latex[length=6pt]}]
\tikzstyle{boldarc}=[draw, bold, arrows={-Latex[length=10pt]}]
\tikzstyle{revarc}=[draw, arrows={Latex[length=6pt]-}]
\tikzstyle{boldrevarc}=[draw, bold, arrows={Latex[length=10pt]-}]

\tikzstyle{bag}=[bagcol, line width=15pt]
\tikzstyle{past}=[draw=gray, fill=pastcol]

\newcommand\tikzvert[1]{
  \raisebox{-.5mm}{
  \hspace{-3mm}
  \begin{tikzpicture}
    \node[vertex, inner sep=.9mm, #1] {};
  \end{tikzpicture}
  \hspace{-4mm}
  }
}

\newcommand{\nextnode}[5][vertex]{\node[small#1] (#2) at ($(#3)+(#4)$) {} edge[revarc,#5] (#3);}

\newcommand{\mybox}[2]{
  \noindent\begin{tikzpicture}
    \node[minimum width=\linewidth-5pt, draw, rounded corners, text width=\linewidth-15pt] (a){#2};
    \node[fill=white, xshift=1em, anchor=west] at (a.north west) {#1};
  \end{tikzpicture}%
}

\newcommand{\myboxprobdef}[6][Question]{
  \label{#6}%
  \noindent\mybox{%
    \ifthenelse{\equal{#4}{}}{}{{#4}\ifthenelse{\equal{#5}{}}{}{ ({#5})}}
  }{%
    \begin{description}
      \item [Input:] {#2}
      \item [#1:] {#3}
    \end{description}
  }
}
\newcommand{\probdef}[5]{
\begin{center}
  \begin{minipage}[t]{.91\textwidth}
    \ifthenelse{\equal{#3}{}}{}{{#3}\ifthenelse{\equal{#4}{}}{}{ ({#4})}}
    \begin{description}
      \item [Input:] {#1}
      \item [Question:] {#2}
    \end{description}
  \end{minipage}
\end{center}
}

\title{Tree Containment Parameterized by Scanwidth\thanks{This work is partially funded by the Dutch Organisation for Scientific Research (NWO) grant OCENW.KLEIN.125 and OCENW.M.21.306.}}

\def\email#1{\href{mailto:#1}{#1}}
\usepackage{authblk}
\author[1]{Leo van Iersel} 
\author[2]{Mark Jones} 
\author[3,4]{Mathias Weller} 
\affil[1]{TU Delft, The Netherlands, \email{l.j.j.vanIersel@tudelft.nl}}
\affil[2]{Middlesex University, London, United Kingdom, \email{m.jones@mdx.ac.uk}}
\affil[3]{LIGM, Université Gustave Eiffel, Marne~la~Vallée, France}
\affil[4]{Centre National de Recherche Scientifique, France, \email{mathias.weller@cnrs.fr}}

\begin{document}
\maketitle

\begin{abstract}
  \TC is a central decision problem in mathematical phylogenetics,
  asking whether a given rooted phylogenetic tree is embeddable in (“displayed by”)
  a given rooted phylogenetic network.
  While the problem is NP-complete for general networks,
  many algorithmic advances have relied on structural parameters that capture how “tree-like” a network is.
  In this paper we investigate \TC under the structural parameter scanwidth,
  a directed width measure generalizing popular parameters measuring tree-likeness of phylogenetic networks.
  We first present a parameterized algorithm that solves the problem in $O(4^{k+k\log k}n + nm^2)$~time,
  where $n$ and $m$ are the numbers of nodes and arcs in the network and $k$ is the width of a given tree-extension.
  Complementing this upper bound, we prove a matching lower bound under the Exponential-Time Hypothesis (ETH),
  showing that there is no algorithm for \TC that runs in $2^{o(c \log c)}n^{O(1)}$~time, even on binary inputs,
  where $c$ is the directed cutwidth of the input network, which upper-bounds the scanwidth~$k$.
\end{abstract}

\section{Introduction}

\paragraph*{Phylogenetics.}
Phylogenetic trees are a central tool for representing the evolutionary relationships between species, genes, languages, or other entities.
%
In many applications, however, purely tree-like models are insufficient.
Evolutionary processes such as hybridization, horizontal gene transfer, or recombination
create reticulate patterns of ancestry that cannot be represented faithfully by a single rooted tree.
To model such phenomena, phylogenetics increasingly relies on \emph{phylogenetic networks},
directed acyclic graphs whose branching structure generalizes rooted phylogenetic trees~\cite{huson2010phylogenetic}.
Such networks have become an important mathematical framework in evolutionary biology and related disciplines
because they allow representation of both vertical inheritance and non-tree-like evolutionary events within a unified combinatorial object.

\paragraph*{Tree Containment.}
A fundamental algorithmic task arising in this context is the \textsc{Tree Containment} problem~\cite{KNL+08,vaniersel2010locating}.
Given a rooted phylogenetic network~$N$ and a rooted phylogenetic tree~$T$ on the same leaf-set,
the problem asks whether~$N$ displays~$T$, that is, whether~$T$ can be obtained from a subgraph of~$N$ by contracting degree-two vertices.
Intuitively, the problem asks whether the evolutionary history represented by~$T$ is compatible with the reticulate history represented by~$N$.
\textsc{Tree Containment} lies at the core of numerous tasks in phylogenetics, including
validation of inferred networks,
comparison of evolutionary hypotheses, and
network reconstruction procedures~\cite{huson2010phylogenetic}.
Consequently, it has been studied extensively from both the algorithmic and complexity-theoretic perspectives~\cite{BW19,KNL+08,vaniersel2010locating,DLS19}.

A natural relaxation of \textsc{Tree Containment}, which can deal with soft polytomies resulting from poorly supported branches in the input,
is the \textsc{Soft Tree Containment} problem, introduced by \citet{BW21}.

\paragraph*{Hardness and Tractable Cases.}
The computational hardness of \textsc{Tree Containment} was first established by \citet{KNL+08}.
This line of research was extended by \citet{vaniersel2010locating},
who provided a systematic complexity classification for several important classes of phylogenetic networks.
Since then, numerous refinements and extensions have appeared~\cite{gambette2016branch,GGL+15,Wel17,GGL+18,Gun18}.
It is also known that counting the number of trees contained in a given network is \#P-complete~\cite{LJS13}.

\paragraph*{Treewidth and Scanwidth.}
Beyond exact polynomial-time solvability on special classes,
parameterized complexity has emerged as a particularly fruitful approach
for coping with the combinatorial explosion inherent in phylogenetic networks~\cite{BW19}.
The guiding idea is to measure how far a network deviates from a tree and to exploit this structure algorithmically.
Indeed, several structural parameters have been investigated in the context of phylogenetic networks.
Classical graph-theoretic parameters such as treewidth and pathwidth naturally suggest themselves
because they capture how close an undirected graph is to a tree.
However, rooted phylogenetic networks are directed, and
many algorithmic properties relevant to phylogenetics depend crucially on the directionality induced by evolutionary time.
As a consequence, undirected width measures may fail to capture the structural properties that make phylogenetic problems tractable.
Motivated by these limitations, recent work introduced \emph{scanwidth},
a width measure specifically tailored to directed acyclic graphs~\cite{scanwidthNPHard}.
Furthermore, scanwidth generalizes and unifies several notions of tree-likeness previously studied in phylogenetics,
such as considering biconnected components individually, leading to the level of a network.

Recent results showed that scanwidth-based algorithms can outperform analogous treewidth-based approaches in parameter dependence.
Indeed, Schestag and Zeh~\cite{SZ26}
recently established the first known complexity-theoretic separation between scanwidth and treewidth.
These developments provide strong evidence that scanwidth captures structural properties of phylogenetic networks
that are fundamentally invisible to classical undirected width measures. Although it is not known whether computing scanwidth is fixed-parameter tractable when the parameter is the scanwidth, an algorithm that runs in polynomial time for fixed scanwidth was shown to be sufficiently fast for applications in phylogenetics~\cite{holtgrefe2026exact}.

\todo[inline]{For Leo: address computing good tree-extensions; mention Niels' “fast enuff in practice” XP-algo, the work-in-progress FPT algo (I cite Norbert's talk in the discussion section)}

\paragraph*{Parameterized Algorithms for Tree Containment.}
Parameterized algorithms for \textsc{Tree Containment} have traditionally focused on measures quantifying the amount of reticulation in the network.
A particularly influential notion in this context is the \emph{level} of a network, which bounds the number of reticulations inside each biconnected component. %
\citet{vaniersel2010locating} already showed that \textsc{Tree Containment} is polynomial-time solvable on level-$k$ networks for fixed~$k$.
Generalizing the notion of level and reticulation-visibility,
\citet{Wel17} showed that \textsc{Tree Containment} can be solved in $O(3^t \cdot |N| \cdot |T|)$~time,
where $t$ is upper-bounded by the maximum number of reticulations that are not visible, over all biconnected components of $N$.
Other parameterized approaches considered reticulation number, treewidth-like parameters, and restrictions to special network classes~\cite{BW19}.

The soft-polytomy version of the problem is fixed-parameter tractable in the combined parameter
consisting of the scanwidth and the maximum degrees of the input phylogenies~\cite{BW26}.
It is conjectured that \textsc{Soft Tree Containment} is W[1]-hard with respect to scanwidth (private communication with Sebastian Bruchhold).

\paragraph*{Our Work.}
Our main contribution is a dynamic programming algorithm that solves \textsc{Tree Containment} in $O(4^{k + k \log k} n + nm^2)$ time,
where~$k$ denotes the width of a given tree-extension of the input network, and~$n$ and~$m$ denote the numbers of vertices and arcs, respectively.
Although previous work already established that \textsc{Tree Containment} is tractable with respect to the \emph{treewidth} of the input network~\cite{vIJW23},
implying tractability also for its scanwidth,
the existing algorithm is much more complex and has a worse dependency on the parameter, which leads to weak performance in practice~\cite{robbert}.
\todo[inline]{For Leo: remove “unimplementable” and mention Robberts implementation of the TC-treewidth algorithm and that it's not practical}
As scanwidth allows us to exploit the directional structure of the network,
we expect the algorithm presented here to perform much better on real-world instances of the problem.
In this way, we add to the argument of \citet{SZ26} that scanwidth is not only of theoretical interest
for problems that are W[1]-hard with respect to treewidth, but may also render problems with known (convoluted) treewidth-based algorithms
accessible for practical use.

We complement our algorithm with a matching lower bound under the Exponential-Time Hypothesis~(ETH)~\cite{IPZ01}.
More precisely, we show that, unless ETH fails, \textsc{Tree Containment} cannot be solved in $2^{o(c \log c)} \cdot n^{O(1)}$~time, with~$c$ the cutwidth of the network. Cutwidth is a measure related to scanwidth similar to the way pathwidth is related to treewidth. Since the scanwidth is always smaller or equal to the cutwidth, our lower bounds holds for scanwidth too.
This establishes that the dependence of our algorithm on the parameter is essentially optimal.
Thus, we provide the first tight complexity characterization of \textsc{Tree Containment} with respect to scanwidth.

\paragraph*{Organization of the Paper.}
The remainder of the paper is organized as follows.
In \cref{sec:prelims}, we introduce the necessary terminology and formal definitions concerning phylogenetic networks, embeddings, and scanwidth.
In \cref{sec:TCFPT}, we present our dynamic programming algorithm.
We then prove the corresponding ETH-based lower bound in \cref{sec:ETHhardness}.
Finally, we conclude with several open questions concerning scanwidth and phylogenetic network algorithms.

\section{Preliminaries and Definitions}\label{sec:prelims}


\paragraph*{DAGs and Reachability.}
We will use the definitions and terminology from~\cite{scanwidthNPHard} wherever possible.
A \emph{(phylogenetic) network} on a set~$L$ is a directed acyclic graph (DAG) containing four types of nodes:
(1)~outdegree-zero nodes (the \emph{leaves}) which are bijectively labelled with elements from~$L$,
(2)~a single node of in-degree zero (the \emph{root}),
\footnote{Note that, sometimes, the root of a phylogenetic network is assumed to have in-degree zero and out-degree at least two.
This does not affect the problems under consideration in this paper as one can always add/remove a root vertex with outdegree one
at the top of each input network.}
%
(3)~nodes with out-degree at least two (the \emph{tree nodes}), and
(4)~nodes with in-degree at least two (the \emph{reticulations}), and
all nodes either have in-degree one or out-degree one.\todo{MW: checked that the degree-1 root is necessary for the definition of TC via pseudo-embeddings.}
If all vertices have total degree at most three, then we say the network is \emph{binary}.
A \emph{(phylogenetic) tree} on a set~$L$ is a phylogenetic network on $L$ with no reticulations.
Relaxing condition (2), we allow a collection of trees with disjoint leaf-sets to be referred to as a \emph{forest} and,
if such a forest can be obtained from a network~$N$ by deleting nodes and arcs, then we call it a \emph{subforest} of~$N$.
In what follows, we will treat leaves in a phylogenetic network as interchangeable with their labels, so the set of leaves \emph{is} $L$.

Let $G$ be a DAG.
We let~$V(G)$ and $E(G)$ denote the set of nodes and arcs of~$G$, respectively.
For any arc~$uv \in A(G)$, we say the $v$ is a \emph{child} of~$u$ and~$u$ is a \emph{parent} of~$v$.
We call~$u$ and~$v$ the \emph{tail} and \emph{head} of~$uv$, respectively.
If $uv$ and $vw$ are arcs in $A(G)$, we also say that~$vw$ is a \emph{child arc} of $v$ and of $uv$,
and similarly $uv$ is a \emph{parent arc} of~$v$ and of~$vw$.
For each vertex~$u$, let~$\outA[G]{u}:=\{uv\mid v\in V(G)\}\cap A(G)$
denote the set of \emph{out-arcs} of $u$ and, similarly, $\inA[G]{v}:=\{uv\mid u\in V(G)\}\cap A(G)$.
If there is a directed path from $u$ to $v$ in $G$,
then we say that $u$ is an \emph{ancestor} of $v$ in $G$, and $v$ is a \emph{descendant} of $u$ in $G$.
Note that a vertex $u$ is considered an ancestor and descendant of itself.
For any directed path $P$ from $u$ to $v$, we call $u$ the \emph{origin} and $v$ the \emph{destination} of $P$, respectively.
Following~\cite{scanwidthNPHard},
we let $\beloeq{G}$ denote the descendant relation in $G$,
that is, 
$v \beloeq{G} u$ if and only if $v$ is a descendant of $u$
(that is, $v$ is reachable from $u$, or $v$ is \emph{below} $u$) in $G$,
and we let $\belo{G}$ denote the result of removing reflexive pairs from $\beloeq{G}$.
Note that the root of a network~$N$ is the unique maximum of~$\beloeq{N}$.
We extend the partial order~$\beloeq{G}$ to $V(G)\cup A(G)$ as follows:
For arcs~$xy$ and~$uv$, we let
\begin{enumerate}[(1)]
  \item $uv \beloeq{G} xy$ if and only if $u \beloeq{G} y$,
  \item $uv \beloeq{G} x$ if and only if $u \beloeq{G} x$, and
  \item $y \beloeq{G} uv$ if and only if $y \beloeq{G} v$.
\end{enumerate}

\paragraph*{Downward-Closed Subforests.}
Let $F$ be a forest and
let $S\subseteq A(F)$ be a non-empty arc-set.
\todo{MW: we need $S\ne\emptyset$ as, otherwise, $F_S$ is empty, so it does not have a degree-1 root, which is mandatory for all networks!}
Let~$F_S$ denote the subgraph of $F$ below~$S$, that is,
the smallest subgraph of $F$ containing the arcs~$\bigcup_{\alpha\in S}\{xy\in A(F)\mid xy\beloeq{G}\alpha\}$.
Then, we call $F_S$ the \emph{forest below~$S$}.
Note that $F_S$ consists of a collection of trees, some of which may share the same root.
This slightly bends established graph-theoretic notions in that two trees with the same root would normally be joined into one tree.
\todo{MW: I added this to address the issue with subforests and degree-1 roots}
If $S$ is an anti-chain, then we call it the \emph{top-arc set of $F_S$}, and a \emph{top-arc set for $L(F_S)$}.
Note that the top-arc set of $F_S$ is unique (it contains exactly the maxima of $A(F_S)$ wrt.~$\belo{F}$),
but the top-arc set for $L(F_S)$ may not be unique.
If, for a subforest~$F'$ of $F$, there is some arc-set~$S$ such that $F'=F_S$, then we call $F'$ \emph{downwards closed} wrt.~$F$,
and we omit the suffix if $F$ is clear from context.
Informally, a downward-closed subforest of $T$ is the subforest consisting of
some set~$S$ of arcs and all their descendants (nodes and arcs according to $\belo{T}$).

For a set~$S\subseteq V(F)$ of non-root nodes of $F$,
we define $F_S$ as $F_{\{uv\in A(F) \mid v\in S\}}$.
This also diverges slightly from established notions as $F_v$ now also contains the arc incoming to $v$ in $F$
(note that $F_S$ is undefined if $S$ contains roots of $F$).

\begin{observation}\label{obs:top-arcs}\label{lem:top-arc-subset}
  Let~$T$ be a tree. Then, the following observations hold.
  \begin{enumerate}[(a)]
    \item $T$ is a downward-closed subforest of~$T$.
    \item Two downward-closed subforests~$F$ and~$F'$ of $T$ have the same top-arcs if and only if $F = F'$.
    \item For all $S\subseteq A(T)$, the forest~$T_S$ is downward-closed.
    \item Any subset of a top-arc set~$S$ is again a top-arc set.
    \item Any $S\subseteq A(T)$ is
      a top-arc set of $T_S$
      if and only if
      $S$ is an anti-chain with respect to~$\belo{T}$.
  \end{enumerate}
\end{observation}

\subsection{Network Width Measures}

This work makes heavy use of width measures of directed acyclic graphs (DAGs),
most prominently their their “cutwidth” and their “scanwidth”,
defined as follows.

\begin{definition}[topological ordering, induced order, cutwidth]
  Let~$G$ be a DAG,
  let~$n:=|V(G)|$, and
  let~$\sigma:V(G) \rightarrow [n]$ be a linear ordering on the vertices of~$G$.
  Then, the order $\beloeq{\sigma}$ \emph{induced} by $\sigma$ is $\{(u,v) \mid \sigma(u)\leq\sigma(v)\}$.
  If, for each~$u,v \in V(G)$, we have~$uv \in A(G)$ only if $\sigma(u) < \sigma(v)$, then $\sigma$ is called a \emph{topological ordering} of $G$.
    
  For each~$i < j$, we let~$\sigma[i\dots j] := \{v \in V(G) \mid i \leq \sigma(v) \leq j\}$.
  For any $r\in [n]$, we let~$C_r(\sigma)$ denote the set of all arcs with one node in $\sigma[1\dots, r]$ and one in $\sigma[r+1,\dots, n]$,
  and we let~$cw_r(\sigma) := |C_r(\sigma)|$.
  Then, the \emph{cutwidth~$cw(\sigma)$ of~$G$ with respect to $\sigma$} is~$\max_r cw_r(\sigma)$ and
  the \emph{(directed) cutwidth of $G$} is the minimum of $cw(\sigma)$ over all (topological) linear orderings~$\sigma$.
\end{definition}

\begin{definition}[Scanwidth, Canonical Tree-Extension, \cite{scanwidthNPHard}]\label{def:scanwidth}
  Let~$G$ be a DAG and
  let~$\Gamma$ be a tree with $V(\Gamma) = V(G)$ such that the relation~$\belo{G}$ is a subset of $\belo{\Gamma}$,
  that is, for all $u,v\in V(G)$, we have $v \belo{G} u$ only if $v \belo{\Gamma} u$.
  Then, $\Gamma$ is called a \emph{tree extension} for $G$.
 
  For each~$v \in V(G)$,
  we let~$GW_v(\Gamma)$ denote the set of arcs~$xy\in A(G)$ for which $x \abov{\Gamma} v \aboveq{\Gamma} y$
  (i.e.~$x$ is strictly above $v$ in $\Gamma$ and $v$ is above (or equal) $y$ in $\Gamma$).
  The \emph{scanwidth of $\Gamma$} is $\max_v |GW_v(\Gamma)|$ and
  the \emph{scanwidth of $G$} is the minimum width of any tree extension for~$G$.
  Finally, if $G[V(\Gamma_v)]$ is weakly connected for all $v\in V(G)$,
  then we call $\Gamma$ \emph{canonical}.
\end{definition}

\begin{observation}[\cite{scanwidthNPHard}, \cite{Holt23}]
  Let $G$ be a DAG.
  Given a tree-extension~$\Gamma'$ of scanwidth~$w$ for $G$,
  a canonical tree-extension $\Gamma$ of scanwidth~$w$ for $G$ can be computed in quadratic time.
  Further,
  \begin{enumerate}[(1)]
    \item $\outdeg[G]{v}\geq\outdeg[\Gamma]{v}$ for each node~$v\in V(G)$, in particular,
    \item the leaf-sets of $G$ and $\Gamma$ are identical.
  \end{enumerate} 
\end{observation}

\noindent
Throughout this work, we will assume that the tree extensions we are working with are canonical.

\citet{scanwidthNPHard} proved that the scanwidth of a DAG~$G$ is at most its directed cutwidth.
It is also known that cutwidth is an upper bound of the pathwidth of the underlying undirected graph of $G$ \cite[Theorem~47]{Bod98}. 
Since directed cutwidth is clearly also an upper bound on cutwidth (any topological ordering is also a linear ordering), we have the following lemma.

\begin{lemma}\cite{scanwidthNPHard,Bod98}
  Let $G$ be a DAG with directed cutwidth $k$. Then the scanwidth and pathwidth of $G$ are at most $k$.
\end{lemma}

\noindent
In~\cref{sec:ETHhardness}, we show lower bounds on the complexity of \textsc{Disjoint Paths} and \TC parameterized by directed cutwidth, based on the Exponential Time Hypothesis (see \cite{IP99,bk-CFK+15}).
As directed cutwidth is an upper bound on scanwidth,
this implies a corresponding bound on \TC parameterized by scanwidth which matches
(up to constants in the exponent) the running time of our algorithm.

\subsection{Tree Containment and Pseudo Embeddings}

Intuitively speaking, the \textsc{Tree Containment} problem asks whether a given tree is contained (that is, “displayed”) by a given network,
where “display” is defined as follows.

\begin{definition}[display]\label{def:display}
  Let $\Nh$ and $\Fg$ be leaf-labelled DAGs.
  Let~$\Fg'$ be a DAG that
  \begin{enumerate}[(a)]
    \item can be derived from~$\Fg$ by repeated subdivision of arcs and
    \item admits a leaf-respecting isomorphism (i.e.~each labelled leaf is mapped to a leaf with the same label) to a subgraph of~$\Nh$.
  \end{enumerate}
  Then, we say that~$\Nh$ \emph{displays}~$\Fg$.
\end{definition}

\noindent
\probdef
{a rooted phylogenetic network~$\Nh$ and rooted phylogenetic tree~$\Tg$}
{Does $\Nh$ display $\Tg$?}
{\TC}{}{}

\noindent
Note that we assume that~$\Nh$ and~$\Tg$ have the same leaf-set as,
otherwise, we can either establish this property by deleting leaves of $\Nh$ (leaves that do not exist in $\Tg$)
or simply rejecting the instance (if $\Tg$ has leaves that do not exist in $\Nh$).



\noindent
We can phrase the \textsc{Tree Containment} problem as
the problem of deciding whether there is a “pseudo-embedding” of $T$ into $N$,
as defined below. 
This is a slight generalization of the common notion of an embedding,
designed to be applicable in our dynamic programming algorithm.



\begin{definition}[pseudo-embedding]\label{def:pseudo-embedding}
  Let $F$ be a downward-closed subforest of $T$.
  A \emph{pseudo-embedding of $F$ into $N$} is a function~$\emb$ that maps all arcs of~$F$ to directed paths of length~$\geq 1$ in $N$ such that
  \begin{enumerate}[(a)]
    \item for arcs $xy,yz \in A(F)$, the destination of the path $\emb(xy)$ is the origin of the path $\emb(yz)$,\label{it:end is start} 
    \item for distinct arcs $uv, xy\in A(F)$, the paths $\emb(uv)$ and $\emb(xy)$ are arc-disjoint, and\label{it:disjoint paths} 
    \item for an arc~$x\ell\in A(F)$ with $\ell \in L(T)$, the path~$\emb(x\ell)$ ends in~$\ell\in L(N)$.\label{it:leaves} 
  \end{enumerate}
  For any arc~$xy\in A(F)$, we denote the destination of $\emb(xy)$ as $\emb(y)$.
  We call $\emb(y)$ the \emph{embedding of $y$}.
\end{definition}

\noindent
Note that the pseudo-embeddings $\emb(xy)$ and $\emb(xz)$ of two arcs~$xy, xz\in A(F)$ may have different origins.
However, this is only possible if the arc incoming to~$x$ in~$T$ is not in~$F$, that is, $xy$ and $xz$ are top-arcs of~$F$.
This does not occur when $F = T$ as $T$ only has a single top arc.
Further note that $\emb(x)$ is not defined if $xy$ is a top arc of~$F$.

\medskip
For ease of presentation, we will let~$\emb$ contain additional mappings that do not figure in \cref{def:pseudo-embedding}.
This allows us to use a pseudo-embedding of $T$ in subforests of $T$ without having to apply a restriction.

\begin{observation}\label{obs:subforest-pseudoemb}
  Let $S\subseteq A(T)$ and $S'\subseteq A(T_S)$.
  Let $\emb$ be a pseudo-embedding of $T_S$.
  Then, $\emb$ is also a pseudo-embedding of $T_{S'}$.%
  \todo{MJ: $\emb$ restricted to $A(T_{S'})$?\\
  MW: we never stated that $\emb$ cannot contain anything that's not used, did we?}
\end{observation}

If a tree $T$ is pseudo-embeddable into a network $N$ with tree-extension~$\Gamma$,
then we can observe that each leaf below any node $y$ in $T$ must also be below the embedding of $y$ in $\Gamma$.

\begin{observation}\label{obs:emb-leaves}
  Let $\Gamma$ be a tree-extension for a network~$N$ and
  let $T$ be a tree.
  Let $S\subseteq A(T)$,
  let~$\emb$ be a pseudo-embedding of $T_S$ into $N$ and
  let~$xy\in A(T_S)$.
  Then, $L(T_y)\subseteq L(\Gamma_{\emb(y)})$.
\end{observation}

\noindent
We observe that paths in the image of a pseudo-embedding do not converge in the same vertex.

\begin{lemma}\label{lem:PEuniquein}
  Let $xz,yz\in A(N)$ with $x\ne y$,
  let $\emb$ be a pseudo-embedding of some~$F$ into $N$,
  let $a,b\in A(F)$ such that $\emb(a)$ contains~$xz$.
  Then, $yz$ is not in~$\emb(b)$.
\end{lemma}
\begin{proof}
  Assume towards a contradiction that $yz$ is an arc of $\emb(b)$.
  Let $u$ and $v$ denote the heads of $a$ and $b$, respectively, and note that $u\ne v$ since $F$ is a forest and $a\ne b$.
  First, if $z\in L(N)$ then, by \cref{def:pseudo-embedding}\ref{it:leaves}, we have $u=z=v$, contradicting~$u\ne v$.
  If $z$ is not a leaf in $N$, then $z$ has a single outgoing arc~$zq$ in~$N$.
  By \cref{def:pseudo-embedding}\ref{it:end is start}, $zq$ is an arc of either~$a$ or a child of $a$ in $F$ and
  the same holds for $b$.
  If $a\ne b$, then $zq$ is in the image of two different arcs in $T$, contradicting \cref{def:pseudo-embedding}\ref{it:disjoint paths}.
  Otherwise, $a=b$ and both $xz$ and $yz$ are in $\emb(a)$, contradicting $\emb(a)$ being a directed path in $N$.
\end{proof}


Now, we can show that \textsc{Tree Containment} is definable using pseudo-embeddings.

\begin{proposition}\label{prop:TC via pseudo-embeddings}
  $N$ displays $T$ if and only if there is a pseudo-embedding of $T$ into $N$.
\end{proposition}
\begin{proof}
  "$\Rightarrow$":
  Suppose that~$N$ displays~$T$, and
  let~$N'$ be the subgraph of~$N$ that is isomorphic to a subdivision~$T'$ of $T$,
  witnessed by an isomorphism~$\iota:V(T') \rightarrow V(N')$.
  Then,
  \begin{enumerate*}[(a)]
    \item $\iota$~is a bijection,
    \item for any~$x,y \in V(T')$, it holds that~$xy \in A(T')$ if and only if $\iota(x)\iota(y) \in A(N')$, and
    \item $\iota(\ell)$~has the same label as~$\ell$ (i.e.~$\iota(\ell)=\ell$ for any $\ell \in L(T')$).
  \end{enumerate*}
  For any arc~$xy \in A(T)$,
  let~$x_0=x,x_1,\dots, x_s = y$ denote the $x$-$y$-path in~$T'$
  (i.e.~$x_1,\dots, x_{s-1}$ are the vertices of $T'$ that subdivide the arc~$xy$ in~$T$).
  Then, define~$\emb(xy)$ to be the path~$\iota(x_0),\iota(x_1),\dots, \iota(x_s)$.
  While it is straightforward to verify that~$\emb$ satisfies the conditions of a pseudo-embedding of~$T$ into~$N$
  (see \cref{def:pseudo-embedding}), we sketch the arguments here for completeness:
  For \cref{def:pseudo-embedding}\ref{it:end is start},
  consider arcs~$xy,yz \in A(T)$ and observe that~$\emb(xy)$ ends in~$\iota(y)$ and~$\emb(yz)$ starts at~$\iota(y)$, by construction.
  For \cref{def:pseudo-embedding}\ref{it:disjoint paths},
  consider distinct arcs $uv,xy \in A(T)$, the $x$-$y$-path and the $u$-$v$-path in $T'$ are arc-disjoint by construction,
  which implies arc-disjointness of the corresponding paths in $N'$, and these paths are exactly $\emb(xy)$ and $\emb(uv)$.
  For \cref{def:pseudo-embedding}\ref{it:leaves},
  consider an arc~$x\ell \in A(T)$ with~$\ell \in L(T)$
  and observe that~$\ell \in L(T')$ and~$\emb(x\ell)$ is, by construction, a path ending in~$\iota(\ell) = \ell$.

  \medskip
  "$\Leftarrow$":
  Suppose that~$\emb$ is a pseudo-embedding of $T$ into $N$.
  Let~$N'$ be the smallest subgraph of~$N$ containing all arcs from paths in~$\emb(xy)$ for any~$xy\in A(T)$, and all vertices incident to these arcs.
  Let~$T'$ be a subdivision of~$T$ such that, for each arc~$xy \in A(T)$, the $x$-$y$-path in~$T'$ has the same length as the path~$\emb(xy)$
  (recall that $\emb$ maps \emph{all arcs} of $T$).
  Now, define the function~$\iota:V(T')\rightarrow V(N')$ as follows:
  for any arc~$xy\in A(T)$,
  consider the $x$-$y$-path $(x_0 = x, x_1, \dots x_s = y)$ in~$T'$,
  let~$\emb(xy)=:(u_0,u_1,\ldots,u_s)$, and
  define~$\iota(x_i) := u_i$ for each~$i \in\{0,1,\dots,s\}$.
  Note that~$\iota$ is well-defined since~$\iota(y)$ can only be ambiguous if $y$ is incident to two or more arcs in $T$,
  but any such~$y$ has at least one in-arc~$xy$ and at least one out-arc~$yz$
  (using the fact that the root and leaves of $T$ have total degree~one)\todo{MW: here, we need that the root has degree 1}
  so, by \cref{def:pseudo-embedding}\ref{it:end is start}, the destination of~$\emb(xy)$ equals the origin of~$\emb(yz)$,
  implying that $\iota(y)$ is not ambiguous.
  %
  To show that $N$ displays $T$, it remains to show that
  \begin{enumerate*}[(1)]
    \item $\iota$ is leaf-respecting,
    \item $\iota$ is injective, and
    \item $xy\in A(T') \iff \iota(x)\iota(y)\in A(N')$ for all~$x,y\in V(T')$.
  \end{enumerate*}
  (see \cref{def:display}).

  \smallskip
  (1):
  Consider a leaf~$\ell\in L(T')=L(T)$ and let~$T$ contain the arc~$x\ell$.
  Then, by \cref{def:pseudo-embedding}\ref{it:leaves}, $\emb(x\ell)$ ends in $\ell$,
  implying that $\iota(\ell)=\ell$.

  \smallskip
  (2):
  Towards a contradiction, assume that there are distinct~$y,y'\in V(T')$ with~$\iota(y) = \iota(y')$
  such that $\iota(y)$ is minimal wrt.~$\belo{N'}$.
  If $y \lneq_{T'} y'$ then, by construction, $\iota(y) \lneq_{N'} \iota(y')$, contradicting $\iota(y)=\iota(y')$.
  Thus, $y$ and $y'$ are incomparable in $T'$.
  In particular, neither $y$ nor $y'$ is the root of $T'$, so they have parents~$x$ and $x'$, respectively.
  By construction, $\iota(x)\iota(y), \iota(x')\iota(y')\in A(N')$.
  As $y$ and $y'$ are incomparable in $T'$, there are distinct arcs~$st, s't'\in A(T)$
  such that $y$ and $y'$ are nodes of the $s$-$t$-path and the $s'$-$t'$-path in $T'$, respectively.
  But then, $\emb(st)$ contains $\iota(x)\iota(y)$ and $\emb(s't')$ contains $\iota(x')\iota(y')$.
  Thus, $\iota(x)\ne\iota(x')$ as, otherwise, $\emb(st)$ and $\emb(s't')$ share the arc $\iota(x)\iota(y)=\iota(x')\iota(y')$,
  contradicting \cref{def:pseudo-embedding}\ref{it:disjoint paths}.
  But now, \cref{lem:PEuniquein} implies that $\iota(x')\iota(y')$ is not in $\emb(s't')$ contradicting our choice of $s't'$.



  \smallskip
  (3):
  Clearly, if $xy\in A(T')$, then $\iota(x)\iota(y)\in A(N')$ by construction.
  For “$\Leftarrow$”, consider an arc~$uv\in A(N')$.
  By~(2), $x:=\iota^{-1}(u)$ and $y:=\iota^{-1}(v)$ are unique nodes in $T'$.
  Since the arc~$uv$ has been constructed, we know that $xy\in A(T')$.
\end{proof}

\section{Tree Containment Parameterized by Scanwidth}\label{sec:TCFPT}


In this section, we show that $\TC$ can be solved in $2^{O(sw\log sw)}n^{O(1)}$~time, where $sw$~denotes the scanwidth of the network.
We first give an overview of the main ideas of our algorithm.
Throughout the remainder of this section,
let~$N$ and $T$ denote the network~$\Nh$ and~$\Tg$ given as input to an instance of \TC,
and let~$\Gamma$ be a canonical tree extension of $N$ with scanwidth $sw$.

Our algorithm is a dynamic programming algorithm that works its way up~$\Gamma$,
starting at the leaves and iteratively working upwards.
Like many treewidth-based dynamic programming algorithms,
our algorithm stores information about partial structures on subgraphs of the input graphs,
that have the potential to be extended to a full solution.
In our case, for each vertex~$v \in V(N)$,
we are interested in finding pseudo-embeddings of certain downward-closed subforests~$F$ of $T$ into the subgraph of~$N$ induced by the vertices of~$\Gamma_v$.
Along with $v$ and $F$, we index these embeddings by a mapping showing which arcs of $F$ are embedded in which arcs of~$GW_v(\Gamma)$.
A crucial observation is that, for each $v$, the number of possible forests~$F$ that need to be considered is bounded by a function of~$sw$.

\subsection{Bounding the number of downward-closed subforests}

In this section we show that there are at most~$4^k$ downward-closed subforests of $T$ with a given leaf set and at most~$k$ top arcs.
As such subforests will be used extensively in our dynamic programming algorithm, this result will be used to bound the running time and memory space required.
To prove this bound, we use the concept of important separators, first introduced in~\cite{Marx2004ParameterizedGS}.

\begin{definition}\cite{CHM2011DirectedMultiwayCut}
  Let $G$ be a directed graph.
  Given two disjoint nonempty sets $X,Y \subseteq V(G)$, we call a set $S \subseteq V(G)\setminus(X\cup Y)$ an \emph{$X - Y$ separator} if there is no path from $X$ to $Y$ in $G\setminus S$ (where $G \setminus S$ denotes the subgraph of $G$ induced by $V(G)\setminus S$).
  We say $S$ is a \emph{minimal} $X - Y$ separator if no proper subset of $S$ is an $X - Y$ separator.
  
  A minimal $X - Y$ separator $S$ is called an \emph{important $X - Y$} separator if there is no $X - Y$ separator $S'$ with $|S'| \leq |S|$ and $R^+_{G\setminus S}(X) \subset R^+_{G\setminus S'}(X)$, where $R^+_A(X)$ is the set of vertices reachable from $X$ in $A$.
\end{definition}

\begin{lemma}\cite{CHM2011DirectedMultiwayCut}\label{thm:importantSeparators}
  Let $X,Y \subseteq V(G)$ be disjoint sets in a directed graph~$G$.
  Then, for every $k \geq 0$ there are at most $4^k$ important $X - Y$ separators of size at most $p$.
  Furthermore, we can enumerate all these separators in time $O(4^k\cdot k(|V(G)+|A(G)|))$.
\end{lemma}

\begin{lemma}\label{lem:subforestBound}
  Let $L'$ be any subset of $L(T)$.
  Then the number of downward-closed subforests of $T$ with at most $k$ top arcs and leaf set $L'$ is at most $4^k$.
\end{lemma} 
\begin{proof}[Proof.\footnotemark{}]\footnotetext{Alternatively, \cref{lem:subforestBound} can be proved using important separators.}
  We show that all top-arc set (which corresponds one-to-one with the downward-closed subforests) of $T$
  can be enumerated with a search-tree with $4^k$ leaf-nodes.
  Let $H$ be the set of maxima, with respect to $\belo{T}$, of the union of all top-arc sets of size at most~$k$ for $L'$.
  Note that $H$ is an anti-chain wrt.~$\belo{T}$ and $|H|\leq k$ as,
  otherwise, for each top-arc set~$S$ for $L'$, there is some arc~$uv\in H$ such that all maximal paths in $T$ that contain $uv$ avoid $S$;
  but, by construction, some top-arc set for~$L'$ contains~$uv$, so there is some leaf~$\ell\in L(T_v)\cap L'$,
  contradicting $S$ being a top-arc set for $L'$.
  Now, as long as $|H|\leq k$ and $H$~contains unmarked arcs, we take the minimum unmarked arc~$h$ of $H$
  with respect to any fix total order of $A(T)$ and recursively either
  \begin{enumerate*}[(a)]
    \item mark~$h$, or
    \item replace~$h$ by the set of child-arcs of $h$ (unless the head of $h$ is a leaf).
  \end{enumerate*}
  In each step, we either mark an arc of $H$ or increase the size of $H$ (since no non-root in $T$ has out-degree one),
  so the corresponding search-tree has depth~$\leq 2k$ and, thus, at most $4^k$ leaf-nodes.
  Further, if all arcs of $H$ are marked in any search-tree leaf, then $H$ forms a top-arc set of size at most $k$ for $L'$.
  It remains to show that this procedure produces all top-arc sets for $L'$ of size at most~$k$.
  Towards a contradiction, assume that some top-arc set~$S$ for~$L'$ of size at most~$k$ is not enumerated by the above process and
  let the size of $S$ be minimum among all such sets.
  
  \textbf{Case 1}: There is some non-leaf~$u\in V(T)$ such that all outgoing arcs of $u$ are in $S$.
  Let $S'$ be the result of replacing all outgoing arcs of $u$ by the unique incoming arc of $u$ in $T$ and
  note that $|S'| < |S|$.
  By minimality of $|S|$, we know that $S'$ was enumerated by the above procedure.
  But then, $S$ is enumerated by applying branching rule~(b) to $S'$, contradicting the fact that $S$ is not enumerated.

  \textbf{Case 2}: No non-leaf of $T$ has all its outgoing arcs in $S$.
  Then, all tails of arcs in $S$ have a path to leaf outside of $L'$ in $T$,
  so all arcs of all top-arc sets for $L'$ are in $T_S$.
  But then, $S=H$ by construction of $H$, so $S$ is enumerated by choosing branch~(a) each time.
\end{proof}

\begin{corollary}
For any $v \in V(N)$, the number of downward-closed subforest of $T$ with at most $k$ top arcs and leaf set $\Gamma_v \cap L(T)$ is at most $4^k$.
\end{corollary}

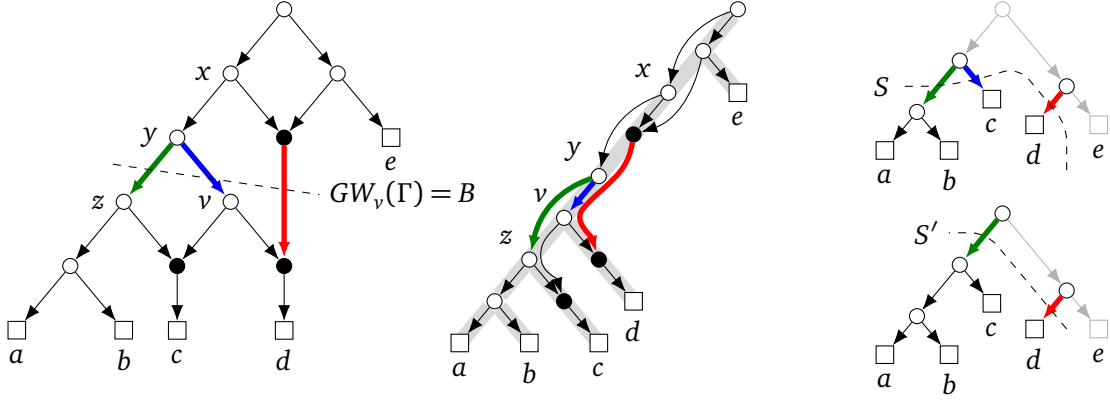
\begin{figure}[t]
  \centering
  \begin{tikzpicture}[xscale=.5, yscale=.6]
    \tikzstyle{B1}=[bold, red]
    \tikzstyle{B2}=[bold, blue]
    \tikzstyle{B3}=[bold, darkgreen]
    \tikzstyle{fade}=[opacity=.3]
    \node[smallvertex] (rt) at (0,0) {};
    \nextnode[vertex, label=left:$x$]{0}{rt}{-135:2}{revarc}
    \nextnode[vertex, label=left:$y$]{00}{0}{-135:2}{revarc}
    \nextnode[vertex, label=left:$z$]{000}{00}{-135:2}{revarc, B3}
    \nextnode{0000}{000}{-135:2}{revarc}
    \nextnode[leaf, label=below:$a$]{l0}{0000}{-135:2}{revarc}
    \nextnode[leaf, label=below:$b$]{l1}{0000}{-45:2}{revarc}

    \nextnode[reti]{0001}{000}{-45:2}{revarc}
    \nextnode[leaf, label=below:$c$]{l2}{0001}{-90:1.42}{revarc}
    
    \nextnode[vertex, label=left:$v$]{001}{00}{-45:2}{revarc, B2}
    \nextnode[reti]{0011}{001}{-45:2}{revarc}
    \nextnode[leaf, label=below:$d$]{l3}{0011}{-90:1.42}{revarc}
    \nextnode[reti]{01}{0}{-45:2}{revarc}
    
    \nextnode{1}{rt}{-45:2}{revarc}
    \nextnode[leaf, label=below:$e$]{l4}{1}{-45:2}{revarc}

    \foreach \u/\v/\s in {001/0001/, 01/0011/B1, 1/01/} \draw[arc, \s] (\u) -- (\v);

    \coordinate (cutoff) at ($(00)!.5!(000)$);
    \draw[dashed] ($(cutoff)+(180-7:1)$) --++(-7:5.5) node[anchor=west] {$GW_v(\Gamma)=B$};

    \node[smallvertex] (Grt) at (12,0) {};
    \foreach [count=\i from 0] \s/\l in {/, /x, reti/, /y, /v, /z, /}
      \node[vertex, small\s, label=135:\ensuremath{\l}] (v\i) at ($(Grt)+(-135:1.3*\i+1.3)$) {};
    \node[smallreti] (r1) at ($(v4)+(-45:1.3)$) {};
    \node[smallreti] (r2) at ($(v5)+(-45:1.3)$) {};
    \foreach \n/\p in {b/v6, c/r2, d/r1, e/v0}{\nextnode[leaf, label=below:$\n$]{Gl\n}{\p}{-45:1.3}{revarc}}
    \nextnode[leaf, label=below:$a$]{Gla}{v6}{-135:1.3}{revarc}
    \foreach \u/\v/\b/\s in {Grt/v0/0/, Grt/v1/30/, v1/v2/0/, v0/v2/-30/, v1/v3/30/, v3/v4/0/B2, v4/r1/0/, v3/v5/30/B3, v5/r2/0/, v5/v6/0/}
    \draw[arc, \s] (\u) to[bend right=\b] (\v);
    \draw[arc, B1] (v2) to[out=-100,in=45] ($(v4)+(.5,.2)$) to[out=-135,in=110] (r1);
    \draw[arc] (v4) to[out=-135,in=90] ($(v5)+(.3,.1)$) to[out=-90,in=110] (r2);
    \begin{pgfonlayer}{background}
      \draw[lightgray, line width=6pt] (v6) -- (Grt);
      \foreach \u/\l in {v0/e, v4/d, v5/c, v6/b, v6/a} \draw[lightgray, line width=6pt] (\u) -- (Gl\l.center);
    \end{pgfonlayer}

    \node[smallvertex, fade] (Trt) at (19,0) {};
    \nextnode{w0}{Trt}{-135:1.6}{revarc, fade}
    \nextnode{w1}{w0}{-135:1.6}{revarc, B3}
    \nextnode{w2}{Trt}{-45:2.4}{revarc, fade}

    \foreach \n/\p/\a/\s/\ns in {a/w1/90//, b/w1/0//, c/w0/0/B2/, d/w2/90/B1/, e/w2/0/fade/fade}{
      \nextnode[leaf, \ns, label=below:$\n$]{Tl\n}{\p}{-45-\a:1.2}{revarc, \s}
    }
    \coordinate (de) at ($(Tld)!.5!(Tle)$);
    \draw[dashed, rounded corners=6] ($(de)+(0,-1)$) -- (de) -- ($(Tlc)+(45:1)$) -- ($(w0)!.5!(Tlc)$) -- ($(w0)!.5!(w1)$) --++(-1,0) node[anchor=east] {$S$};

    \node[smallvertex] (Qrt) at (19,-4.5) {};
    \nextnode{q0}{Qrt}{-135:1.6}{revarc, B3}
    \nextnode{q1}{q0}{-135:1.6}{revarc}
    \nextnode{q2}{Qrt}{-45:2.4}{revarc, fade}

    \foreach \n/\p/\a/\s/\ns in {a/q1/90//, b/q1/0//, c/q0/0//, d/q2/90/B1/, e/q2/0/fade/fade}{
      \nextnode[leaf, \ns, label=below:$\n$]{Ql\n}{\p}{-45-\a:1.2}{revarc, \s}
    }
    \draw[dashed, rounded corners=6] ($(Qld)!.5!(Qle)$) --++(135:3) --++(-1,0) node[anchor=east] {$S'$};

  \end{tikzpicture}
  \caption{%
  Illustration of signatures $(v,S,\psi,\outdeg[\Gamma]{v})$ and $(v,S',\psi',\outdeg[\Gamma]{v})$,
  where
  $\psi$ and $\psi'$
  map each of the colored arcs marked as $S$ (top right) and $S'$ (bottom right), respectively, to the arc of the same color in $GW_v(\Gamma)$ (left, middle).
  Note that $S'$ does not contain blue arcs, so $\psi'$ does not map blue arcs.
  {\bf Left}: input network~$N$ with arcs of $GW_v(\Gamma)$ in bold and color.
  {\bf Middle}: tree extension~$\Gamma$ (thick gray) into which the arcs of~$N$ have been drawn.
  {\bf Right}: two copies of the input tree $T$ together with the top-arc sets~$S$ (top) and~$S'$ (bottom) in bold,
  with their corresponding downward-closed forests (non-faded arcs).
  Note that both $\phi$ and $\phi'$ are valid.
  }
  \label{fig:DP table}
\end{figure}

\subsubsection{Algorithm Presentation}

%
Unsurprisingly, our algorithm is a bottom-up dynamic programming on the given canonical tree-extension~$\Gamma$,
which we state in the form of a recursive algorithm to which memoization is to be applied
to avoid recomputing previously computed results.
Given a node~$v$ of the network with “bag”~$GW_v(\Gamma)$ of size~$k$,
as well a top-arc set~$S$ of size at most~$k$ for $L(\Gamma_v)$ in $T$ and
an appropriate (no two arcs in $\psi(S)$ have the same head) injection $\psi:S\to GW_v(\Gamma)$,
we distribute all arcs of $GW_v(\Gamma)$, that do not have~$v$ as their head, onto the children of~$v$ in~$\Gamma$.
Then, for the unique arc~$uv\in GW_v(\Gamma)\cap\psi(S)$ (if it exists), we distinguish the cases that the embedding of the arc of $T$
that uses $uv$ has $v$ as its destination or not (see \cref{fig:DP}).
In the first case, we distribute the child-arcs of $\psi^{-1}(uv)$ onto the bags of the children of $v$ in $\Gamma$.
In the second case, we “continue” the embedding of $\psi^{-1}(uv)$ in the appropriate child of $v$ in $\Gamma$.
To support networks with high out-degree, we distribute arcs using a dynamic programming over the children~$q_j$ of $v$ in $\Gamma$
that enforces baseline properties for $j=0$ and, at each child~$q_j$, picks up all arcs~$xy\in S$ with the correct leaf-set,
that is, with $L(T_y)\subseteq L(\Gamma_{q_j})$.
Intuitively,
if it is possible to distribute all arcs in $S$
onto the children of~$v$ in $\Gamma$ without violating this leaf-constraint,
then we call~$S$ “$v$-splittable”.

\begin{definition}\label{def:v-split}
  Let~$v$ be a node of $N$ with children~$q_1,q_2,\ldots$ in $\Gamma$ and
  let~$Y\subseteq A(T)$.
  Let~$Y_i=\{xy\in Y \mid L(T_y)\subseteq L(\Gamma_{q_i})\}$ for all~$1\leq i\leq\outdeg[\Gamma]{v}$.
  Then, we call $(Y_1,Y_2,\ldots)$ a \emph{$v$-split} of $Y$.
  Further, if $Y = \bigcup_i Y_i$, then we say that $Y$~is \emph{$v$-splittable} into~$(Y_1,Y_2,\ldots)$.
  If $v$ is a leaf in $\Gamma$, we define~$\emptyset$ as $v$-splittable.
\end{definition}

\noindent
Note
that the sets~$Y_i$ form a partition of $Y$ and
that $\emptyset$ admits the $v$-split $(\emptyset,\emptyset,\ldots)$ for all~$v$.

Our algorithm takes as input a function~$\psi$ mapping top-arc sets~$S$ in $T$ into $A(N)$.
We~call an injective function~$\psi:S\to A(N)$ \emph{valid} if, for all~$xy\in S$ and $wz := \psi(xy)$, we have $L(T_y)\subseteq L(\Gamma_z)$.
Note that validity is hereditary, that is, all restrictions of valid functions are themselves valid.
In the following, we let~$\replace{\psi}{xy}{vw}$ denote the result of re-pointing~$xy$ to~$vw$ or,
more formally, $\replace{\psi}{xy}{vw}:=(\psi\setminus (xy,\psi(xy)))\cup\{(xy,vw)\}$.


\begin{definition}\label{def:witness}\label{def:signature}
  Let~$v$ be a node of~$N$.
  let~$q_1,\dots q_c$ be the children of $v$ in $\Gamma$ (possibly $c=0$), and
  let~$j\leq c$.
  Let~$S$ be a top-arc set of $L(\Gamma_v)$,\todo{MW: size at most~$|GW_v(\Gamma)|$?}
  let~$\psi:S\to GW_v(\Gamma)$ be a valid function.
  Then, we call the tuple~$(v,S,\psi,j)$ a \emph{signature}
  and we define~$\chi[v,S,\psi,j]:=\textsc{True}$
  if and only if
  \begin{enumerate}[(a)]
    \item\label{it:j0}
      $j=0$ and $S$ is $v$-splittable or $S\setminus\{xy\}$ is $v$-splittable for some $xy\in S$ with $\psi(xy)\in\inA[N]{v}$.
    \item\label{it:v-split}
      $j>0$, $S$ is $v$-splittable into $(S_1,S_2,\ldots)$ and
      there is a pseudo-embedding~$\emb$ of $T_{\bigcup_{i\leq j}S_i}$ into~$N$ such that,
      for each~$wz \in \bigcup_{i\leq j}S_i$, the path~$\emb(wz)$ begins with~$\psi(wz)$ (see \cref{fig:DP}(left, middle)), or
    \item\label{it:S-xy v-split}
      $j>0$ and there is some~$xy\in S$ with $\psi(xy)\in\inA[N]{v}$ such that
      $S\setminus\{xy\}$ is $v$-splittable into $(S_1,S_2,\ldots)$ and
      $\outA[T]{y}$ is non-empty and $v$-splittable into~$(Y_1,Y_2,\ldots)$ and
      there is a pseudo-embedding~$\emb$ of $T_{\bigcup_{i\leq j}S_i\cup Y_i}$ into~$N$ such that,
      for each~$wz \in \bigcup_{i\leq j}S_i$, the path~$\emb(wz)$ begins with~$\psi(wz)$ and
      for each~$yz \in \bigcup_{i\leq j}Y_i$, the startpoint of $\emb(yz)$ is $v$
      (see \cref{fig:DP}(right)).
  \end{enumerate}
  If $j>0$, then we call $\emb$ a \emph{witnessing pseudo-embedding} (or \emph{witness}) for $(v,S,\psi,j)$.
\end{definition}

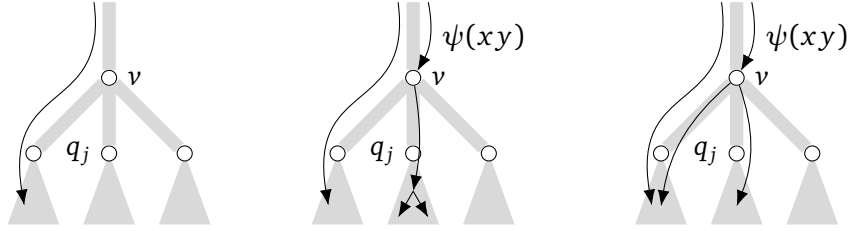
\begin{figure}[t]
  \centering
  \begin{tikzpicture}
    \node[smallvertex, label=right:$v$] (v) at (0,0) {} edge[lightgray, gamma] ($(v)+(90:1)$);
    \foreach[count=\i from 0] \l in {,q_j,} {
      \coordinate (q\i) at ($(v)+(-1+\i,-1)$);
      \path[fill=lightgray] (q\i.center) -- ($(q\i)+(-70:1)$) -- ($(q\i)+(-110:1)$) -- cycle;
      \node[smallvertex, label=left:\ensuremath{\l}] at (q\i) {} edge[lightgray, gamma] (v);
    }
    \draw[arc] ($(v)+(90:1)-(.2,0)$) .. controls ($(v)+(0,-.5)$) and ($(q0)+(-.4,1)$) .. ($(q0)+(-100:.7)$);
  \end{tikzpicture}
  \hspace{6ex}
  \begin{tikzpicture}
    \node[smallvertex, label=right:$v$] (v) at (0,0) {} edge[lightgray, gamma] ($(v)+(90:1)$);
    \foreach[count=\i from 0] \l in {,q_j,} {
      \coordinate (q\i) at ($(v)+(-1+\i,-1)$);
      \path[fill=lightgray] (q\i.center) -- ($(q\i)+(-70:1)$) -- ($(q\i)+(-110:1)$) -- cycle;
      \node[smallvertex, label=left:\ensuremath{\l}] at (q\i) {} edge[lightgray, gamma] (v);
    }
    \draw[arc] ($(v)+(90:1)+(.2,0)$) to[bend left=20] node[midway, xshift=4ex]{$\psi(xy)$} (v);
    \coordinate (q) at ($(q1)+(-90:.5)$);
    \draw[arc] (v) to[bend left=10] (q);
    \draw[arc] (q) to ($(q)+(-60:.4)$);
    \draw[arc] (q) to ($(q)+(-120:.4)$);
    \draw[arc] ($(v)+(90:1)-(.2,0)$) .. controls ($(v)+(0,-.5)$) and ($(q0)+(-.4,1)$) .. ($(q0)+(-100:.7)$);
  \end{tikzpicture}
  \hspace{6ex}
  \begin{tikzpicture}
    \node[smallvertex, label=right:$v$] (v) at (0,0) {} edge[lightgray, gamma] ($(v)+(90:1)$);
    \foreach[count=\i from 0] \l in {,q_j,} {
      \coordinate (q\i) at ($(v)+(-1+\i,-1)$);
      \path[fill=lightgray] (q\i.center) -- ($(q\i)+(-70:1)$) -- ($(q\i)+(-110:1)$) -- cycle;
      \node[smallvertex, label=left:\ensuremath{\l}] at (q\i) {} edge[lightgray, gamma] (v);
    }
    \draw[arc] ($(v)+(90:1)+(.2,0)$) to[bend left=20] node[midway, xshift=4ex]{$\psi(xy)$} (v);
    \draw[arc] (v) to[bend left=20] ($(q1)+(-90:.7)$);
    \draw[arc] (v) to[bend left=-20] ($(q0)+(-90:.7)$);
    \draw[arc] ($(v)+(90:1)-(.2,0)$) .. controls ($(v)+(0,-.5)$) and ($(q0)+(-.4,1)$) .. ($(q0)+(-100:.7)$);
  \end{tikzpicture}
  \caption{Illustration of the three main cases occurring in the algorithm.
    $\Gamma$ is depicted in gray, and parts of an embedding of $T$ are depicted in solid black.
    Left: no arc of $S$ is mapped to a path containing~$v$ by the embedding.
    Middle: $xy\in S$ is mapped to a path of which~$v$ is an inner node.
    Right: $xy\in S$ is mapped to a path ending in~$v$.
  }
  \label{fig:DP}
\end{figure}

\noindent
For $j=\outdeg[\Gamma]{v}$, \cref{def:witness} collapses to the following, simpler characterization, which we use in the following.

\begin{lemma}\label{lem:last-entry}
  Let~$(v,S,\psi,\outdeg[\Gamma]{v})$ be a signature.
  Then,
  $\chi[v,S,\psi,\outdeg[\Gamma]{v}]=\textsc{True}$
  if and only if
  there is a pseudo-embedding~$\emb$ of $T_S$ into $N$ such that $\emb(a)$ begins with $\psi(a)$ for all~$a\in S$.
\end{lemma}
\begin{proof}
  “$\Rightarrow$”:
  Let~$j:=\outdeg[\Gamma]{v}$.
  We consider the three cases of \cref{def:witness}.

  \textbf{Case~(1)}:
  $\outdeg[\Gamma]{v}=0$, $v$ is a leaf, $S=\{xy\}$ and $v=y$.
  But then, $\emb=\psi$ is a pseudo-embedding of $T_{xy}$ into $N$ and $\emb(xy)$ begins with $\psi(xy)$.

  \textbf{Case~(2)}:
  $j>0$ and $S$ is $v$-splittable into~$(S_1,S_2,\ldots)$ and
  there is a pseudo-embedding~$\emb'$ of $T_{\bigcup_i S_i}$ into~$N$ and,
  for each~$wz\in\bigcup_i S_i$, the path~$\emb'(wz)$ begins with $\psi(wz)$.
  But since~$S=\bigcup_i S_i$, the claim follows for $\emb=\emb'$.

  \textbf{Case~(3)}:
  $j>0$ and there is some~$xy\in S$ with $\psi(xy)\in\inA[N]{v}$ such that
  $S\setminus\{xy\}$ is $v$-splittable into $(S_1,S_2,\ldots)$ and
  $\outA[T]{y}$ is non-empty and $v$-splittable into~$(Y_1,Y_2,\ldots)$ and
  there is a pseudo-embedding~$\emb'$ of $T_{\bigcup_i S_i\cup Y_i}$ into~$N$ such that
  for each~$wz \in \bigcup_i S_i$, the path~$\emb'(wz)$ begins with~$\psi(wz)$ and
  for each~$yz \in \bigcup_i Y_i$, the startpoint of $\emb'(yz)$ is $v$.
  We show that $\emb:=\emb'\cup\{(xy,\psi(xy))\}$ satisfies the claim
  by proving the conditions of \cref{def:pseudo-embedding} for $xy$
  (clearly, $\emb(a)$ begins with $\psi(a)$ for all $a\in S$).
  First, for all~$yz\in\outA[T]{y}$, we have that $\emb(yz)$ has origin~$v$, which is also the destination of $\emb(xy)$, proving \cref{def:pseudo-embedding}\ref{it:end is start}.
  Second, for all~$wz\in A(T_S)\setminus\{xy\}$, 
  every arc in $u'v'$ in the path $\emb(wz)$ satisfies either $v' \leq_\Gamma \psi(a)$ for some $a \in \bigcup_{i \leq j} S_i$ (if $wz \in A(T_{\bigcup_{i \leq j} S_i})$) or $v' <_{\Gamma} v$ (if $wz \in T_{\bigcup_{i \leq j} Y_i}$).
  In either case $u'v' \neq uv$, and so 
  $\emb(wz)$ and $\emb(xy)$  are arc-disjoint, proving \cref{def:pseudo-embedding}\ref{it:disjoint paths}.
  Third, $y\notin L(T)$ since $\outA[T]{y}$ is non-empty.

  \medskip
  “$\Leftarrow$”:
  First, suppose that $v$ is a leaf, that is, $\outdeg[\Gamma]{v}=0$.
  Then, since $L(T_S)=L(\Gamma_v)=\{v\}$ and $T$ is a tree,
  $S$ contains a single arc~$xy$ and $v=y$.
  Further, since $\emb(xy)$ begins with $\psi(xy)$, we have $\psi(xy)\in\inA[N]{v}$.
  Then, the claim follows by \cref{def:witness}\ref{it:j0}.

  Second, suppose that $v$ is not a leaf, but $S$ is $v$-splittable into $(S_1,S_2,\ldots)$.
  Then, $S=\bigcup_i S_i$ so $\emb$ is a pseudo-embedding of $T_{\bigcup_i S_i}$ into $N$ such that,
  for each $wz\in\bigcup_i S_i$, the path $\emb(wz)$ begins with $\psi(wz)$.
  Then, the claim follows by \cref{def:witness}\ref{it:v-split}.

  Third, suppose that $v$ is not a leaf and $S$ is not $v$-splittable.
  Then, $v$ has children~$q_1,q_2,\ldots$ in $\Gamma$.
  For each~$i$, let~$S_i:=\{wz\in S\mid L(T_z)\subseteq L(\Gamma_{q_i})$.
  Since $S$ is not $v$-splittable, $S\setminus \bigcup_i S_i$ contains an arc~$xy$.
  Then, (A)~$\psi(xy)\in\inA[N]{v}$ since, otherwise, $\emb(y)\belo{\Gamma} q_i$ for some~$i$,
  implying~$L(T_y) \stackrel{Obs.~\ref{obs:emb-leaves}}{\subseteq} L(\Gamma_{\emb(y)})\subseteq L(\Gamma_{q_i})$,
  contradicting~$xy\notin S_i$.
  Analogously, $v$ is the destination of $\emb(xy)$.
  Since $v$ is not a leaf, \cref{def:pseudo-embedding}\ref{it:leaves} implies that $y$ is not a leaf.
  Further, by \cref{def:pseudo-embedding}\ref{it:end is start},
  for each $yz\in\outA[T]{y}$, (B)~the startpoint of $\emb(yz)$ is $v$,
  implying that (C)~$\outA[T]{y}$ is non-empty and $v$-splittable into $(Y_1, Y_2, \ldots)$.
  By \cref{lem:PEuniquein}, $xy$ is the only arc in $S\setminus\bigcup_i S_i$,
  so (D)~$S\setminus\{xy\}$ is $v$-splittable into~$(S_1,S_2,\ldots)$.
  Finally, since all arcs in 
  $\bigcup_{i\leq j}S_i\cup Y_i$
  are in $T_S$, \cref{obs:subforest-pseudoemb} implies that
  (E)~$\emb$ is a pseudo-embedding for
  $T_{\bigcup_{i\leq j}S_i\cup Y_i}$
  and $\emb(a)$ begins with $\psi(a)$ for all 
  $a \in \bigcup_{i\leq j}S_i$.
  %
  Together, (A), (D), (C), (E), and (B) imply \cref{def:witness}\ref{it:S-xy v-split}.
\end{proof}

\SetKwFunction{Contained}{Contained}
\begin{algorithm}[t!] 
  \SetKwFunction{FName}{MyFunctionName} 
  \caption{A recursive function checking whether the subforest rooted at the top-arc set~$S$ is contained in the subnetwork below $v$ in $\Gamma$,
    where arcs of $S$ are mapped to arcs in $GW_{v}(\Gamma)$ by $\psi$. The last argument $j$ indicates how many children of $v$ in $\Gamma$
    have already been processed, that is, those children are unavailable to display subtrees below~$S$.}
  \label{alg:compute-chi}
  \textbf{function }\Contained{$v$, $S$, $\psi$, $j$}:\\
    \If{$j=0$}{
      \Return{($S$ is $v$-splittable) or ($\exists_{xy\in S}\;\psi(xy)\in\inA[N]{v}$ and $S\setminus\{xy\}$ is $v$-splittable)}\;
    }
    \lIf{not \Contained{$v$,$S$,$\psi$,$j-1$}}{\Return{false}}
    $q_j\gets j^\text{th}$ child of $v$ in $\Gamma$\;
    $S_j \gets \{wz\in S\mid L(T_z)\subseteq L(\Gamma_{q_j})\}$\;
    $\psi_j\gets $ restriction of $\psi$ to $S_j$\;
    \If(\tcp*[f]{don't recheck this for each $j$}){$\exists_{xy\in S}\; \psi(xy)\in\inA[N]{v}$}{
      \If(\tcp*[f]{don't recheck this for each $j$}){$\outA[T]{y}$ is $v$-splittable into $(Y_1,Y_2,\ldots)$}{
        \If(\tcp*[f]{\cref{fig:DP} (middle)}){$Y_j = \outA[T]{y}$}{
          \ForEach{$vw\in\outA[N]{v}$ with $L(T_y)\subseteq L(\Gamma_w)$}{
            \lIf{\Contained{$q_j,S_j,\replace{\psi_j}{xy}{vw},\outdeg[\Gamma]{q_j}$}}{\Return{true}}
          }
        }
        \If(\tcp*[f]{\cref{fig:DP} (right)}){$\outA[T]{y}\ne\emptyset$}{
          \ForEach{valid $\psi'_j:Y_j\to \outA[N]{v}\cap GW_{q_j}(\Gamma)\setminus\psi_j(S_j)$}{
            \lIf{\Contained{$q_j,(S_j\setminus\{xy\})\cup Y_j,\psi_j\cup\psi'_j,\outdeg[\Gamma]{q_j}$}}{\Return{true}}
          }
        }
      }\lElse{\Return{false}}
    }\lElse(\tcp*[f]{\cref{fig:DP} (left)}){\Return{\mbox{Contained($q_j,S_j,\psi_j,\outdeg[\Gamma]{q_j}$)}}}
\end{algorithm}

\subsubsection{Correctness}
\Cref{alg:compute-chi} shows how to compute $\chi[v,S,\psi,j]$ recursively for any signature.
In the following, we show its correctness.
To this end, recall that the leaves of $\Gamma$ are exactly the leaves of $N$ and
recall that $GW_v(\Gamma)$ is the set of all arcs $uw$ in $N$ for which $w\beloeq{\Gamma} v \belo{\Gamma} u$.

The following lemma shows that all recursive calls made by our algorithm correspond to signatures,
which will be important for the correctness and the running-time considerations.
We remark that, for all signatures~$(v,S,\psi,j)$ with~$j>0$, the tuple~$(v,S,\psi,j-1)$ is a signature.

\begin{lemma}\label{lem:sig}
  Let~$(v,S,\psi,j)$ be a signature with $j>0$ and
  let~$q_j$ be a child of $v$ in $\Gamma$.
  Let~$S_j:=\{wz\in S \mid L(T_z)\subseteq L(\Gamma_{q_j})\}$ and
  let~$\psi_j$ be the restriction of $\psi$ to $S_j$.
  Let $Z:=\{xy\in S \mid \psi(xy)\in\inA[N]{v}\}$.

  If $Z=\emptyset$, then $(q_j,S_j,\psi_j,\outdeg[\Gamma]{q_j})$ is a signature.
  Otherwise,
  let~$Z=\{xy\}$ for some~$xy\in S$,
  let~$Y:=\outA[T]{y}$ be $v$-splittable into~$(Y_1,Y_2,\ldots)$,
  let~$\psi'_j:Y_j\to GW_{q_j}(\Gamma)$ be a valid function, and
  let~$vw\in\outA[N]{v}$ with $L(T_y)\subseteq L(\Gamma_w)$.
  If $Y_j=Y$, then $(q_j,S_j,\replace{\psi_j}{xy}{vw},\outdeg[\Gamma]{q_j})$ is a signature.
  Also,
  if $Y_j\ne\emptyset$, then $(q_j,(S_j\setminus Z)\cup Y_j,\psi_j\cup\psi'_j,\outdeg[\Gamma]{q_j})$ is a signature.
\end{lemma}
\begin{proof}
  By \cref{def:witness}, for each of the three claims, we need to show that
  (a)~the second argument is a top-arc set of $L(\Gamma_{q_j})$ and
  (b)~the third argument is a valid function from the second argument to $GW_{q_j}(\Gamma)$.

  (a):
  Since $S_j\subseteq S$, \cref{lem:top-arc-subset} implies that~$S_j$ is a top-arc set.
  By construction, $T_{S\setminus Z}$ and $T_Y$ are arc-disjoint and $Y$~is a top-arc set,
  so $(S\setminus Z)\cup Y$ and its subset~$(S_j\setminus Z)\cup Y_j$ are top-arc sets.
  It remains to show that they are top-arc sets \textbf{for} $L(\Gamma_{q_j})$,
  that is,
  (1)~$L(T_{S_j})=L(\Gamma_{q_j})$ if $Z=\emptyset$ or $Y_i=Y$ and
  (2)~$L(T_{(S_j\setminus Z)\cup Y_j})=L(\Gamma_{q_j})$ if $Y_j\ne\emptyset$.
  
  “$\subseteq$”: For each~$\ell\in L(T_{S_j\cup Y_j})$ (where $Y_j=\emptyset$ is possible),
  there is some arc~$xy\in S_j\cup Y_j$ with $\ell\beloeq{T} y$ and,
  by definition of $S_j$ and $Y_j$, we have $\ell\in L(T_y)\subseteq L(\Gamma_{q_j})$.
  
  “$\supseteq$”:
  Consider some~$\ell\in L(\Gamma_{q_j})\subseteq L(\Gamma_v)$ and
  let~$p$ be some root-$\ell$-path in $N$ containing~$q_j$.
  Then, $p$ contains an arc~$xy\in S$, since $S$ is a top-arc set of $L(\Gamma_v)$.
  If~$xy\notin S_j\cup Z$, then $\ell\in L(\Gamma_{q_k})$ for some child~$q_k\ne q_j$ of $v$ in $\Gamma$,
  so $\ell\beloeq{\Gamma} q_k,q_j$, contradicting $\Gamma$ being a tree.
  If~$xy\in S_j$, then $\ell\in L(T_{S_j})$, proving~(1) and~(2).
  Thus, suppose~$xy\in Z$ in the following.
  First, suppose~$Y\ne\emptyset$.
  Then, $p$~contains some~$yz\in\outA[T]{y}=Y$ and $yz\in Y_j$ since $(Y_1,Y_2,\ldots)$ is a $v$-split of $Y$,
  implying~$\ell\in L(T_{Y_j})\subseteq L(T_{(S_j\setminus Z)\cup Y_j})$ and proving~(2) in this case.
  If, additionally, $Y_j=Y$, then $L(T_y)\subseteq L(\Gamma_{q_j})$, so $xy\in S_j$,
  implying~$\ell\in L(T_{S_j})$ and proving~(1) in this case.
  Second, suppose~$Y_j=Y=\emptyset$, that is, $y$ is a leaf of $T$.
  Then, $L(T_y) = \{y\}\subseteq L(\Gamma_{q_j})$, so $xy\in S_j$, implying $\ell\in L(T_{S_j})$ and proving~(1) and~(2) in this case.

  (b):
  First, note that $\psi_j$ is valid since validity is hereditary.
  Second, if $Y_j=Y=\outA[T]{y}$, then $\replace{\psi_j}{xy}{vw}$ is valid since $L(T_y)\subseteq L(\Gamma_w)$ in this case.
  Third, $\psi_j\cup\psi'_j$ is valid if both $\psi_j$ and $\psi'_j$ are valid.
  It remains to show that the third argument maps into~$GW_{q_j}(\Gamma)$.
  Towards a contradiction, assume that there is some arc~$a$ that gets mapped to some~$b\notin GW_{q_j}(\Gamma)$.
  Since $\psi$ is valid, either $\psi$~does not map~$a$ (implying~$a\notin S$) or $\psi(a)\ne b$.
  But then, $Z\ne\emptyset$ since, otherwise $b=\psi_j(a)=\psi(a)\ne b$, so $Z=\{xy\}$.
  But then,
  (1)~$b=\psi_j(a)=\psi(a)\ne b$ or
  (2)~$b=\psi'_j(a)$ and $a\in Y_j$, but $b=\psi'_j(a)\in\psi'_j(Y_j)\subseteq GW_{q_j}(\Gamma)$, contradicting $b\notin GW_{q_j}(\Gamma)$.
  It remains to cover the case that~$Z=\{xy\}$ and~$Y_j=Y$.
  If $a\ne xy$, then $b=\replace{\psi_j}{xy}{vw}(a)=\psi(a)\ne b$ and,
  if $a=xy$, then $b=\replace{\psi_j}{xy}{vw}(xy)=vw$.
  Since~$Y_j=Y$, we have~$L(T_y)\subseteq L(\Gamma_{q_j})$ and since $L(T_y)\subseteq L(\Gamma_w)$, we know that $w$ and $q_j$ are comparable in $\Gamma$.
  Further $w\belo{N} v$, implying $w\beloeq{\Gamma} q_j$, so $vw=b\in GW_{q_j}(\Gamma)$ contradicting $b\notin GW_{q_j}(\Gamma)$.
\end{proof}

\noindent
Note that correctness is clear if $j=0$, so we focus on the case that~$j>0$.
We split the correctness proof of \Cref{alg:compute-chi} into the cases
that $\psi$ maps some arc~$xy$ of $S$ into~$\inA[N]{v}$ (\Cref{fig:DP}(middle,right)) or not (\Cref{fig:DP}(left)).

\begin{lemma}\label{lem:DP_correct1}
  Let~$(v,S,\psi,j)$ be a signature with $j>0$ and
  let~$q_1,q_2,\dots$ be the children of $v$ in $\Gamma$.
  Let~$S$ be $v$-splittable into $(S_1,S_2,\ldots)$,
  let~$\psi_j$ be the restriction of $\psi$ to $S_j$, and
  let~$\psi(S)\cap\inA[N]{v}=\emptyset$.

  Then, $\chi[v,S,\psi,j]=\textsc{True}$
  if and only if
  $\chi[v,S,\psi,j-1]=\textsc{True}$ and
  $\chi[q_j,S_j,\psi_j,\outdeg[\Gamma]{q_j}]$.
\end{lemma}
\begin{proof}
  “$\Rightarrow$”:
  By \cref{def:witness}\ref{it:v-split}, there is some pseudo-embedding $\emb$ of $T_{\bigcup_{i\leq j} S_i}$ into $N$ such that,
  for each $wz\in\bigcup_{i\leq j} S_i$, the path~$\emb(wz)$ begins with the arc~$\psi(wz)$.
  Then, by \cref{obs:subforest-pseudoemb}, $\emb$ is also a pseudo-embedding of $T_{\bigcup_{i\leq j-1} S_i}$,
  so $\emb$ satisfies \cref{def:witness}\ref{it:v-split} for $j-1$, implying $\chi[v,S,\psi,j-1]=\textsc{True}$.
  Furter, by \cref{lem:sig}, $(q_j,S_j,\psi_j,\outdeg[\Gamma]{q_j})$ is a signature, and
  by \cref{obs:subforest-pseudoemb}, $\emb$ is also a pseudo-embedding for $T_{S_j}$ and,
  for each $xy\in S_j$, the path~$\emb(xy)$ begins with $\psi(xy)=\psi_j(xy)$, so \cref{lem:last-entry} implies the claim.

  \medskip
  “$\Leftarrow$”:
  Let $\emb'$ be a witness for $(v,S,\psi,j-1)$, that is,
  $\emb'$ is a pseudo-embedding of $T_{\bigcup_{i\leq j-1} S_i}$ into $N$ such that,
  for each $xy\in\bigcup_{i\leq j-1} S_i$, the path~$\emb'(xy)$ begins with $\psi(xy)$.
  Further, by \cref{lem:last-entry}, there is a pseudo-embedding~$\emb_j$ for $T_{S_j}$ into $N$ such that,
  for all~$xy\in S_j$, the path~$\emb_j(xy)$ begins with $\psi_j(xy)$.
  We claim that $\emb:=\emb'\uplus \emb_j$ is a witness for $(v,S,\psi,j)$, that is,
  $\emb$ is a pseudo-embedding of $T_{\bigcup_{i\leq j} S_i}$ into $N$ such that,
  for each $i\leq j$ and each $a\in S_i$, the path~$\emb(a)$ begins with $\psi(a)$.
  First, \cref{def:pseudo-embedding}\ref{it:end is start} and \ref{it:leaves} follow from the fact that $\emb'$ and $\emb_j$ are pseudo-embeddings.
  Second, assume that \cref{def:pseudo-embedding}\ref{it:disjoint paths} is false, that is,
  there is some~$i\leq j-1$ and arcs~$xy$ and $wz$ in $T_{S_i}$ and $T_{S_j}$, respectively, such that $\emb(xy)$ and $\emb(wz)$ share an arc~$st$ in $N$.
  But then, $t \beloeq{N} q_i, q_j$, so $t\beloeq{\Gamma} q_i,q_j$, contradicting $q_i$ and $q_j$ being incomparable in the tree~$\Gamma$.
  Thus, $\emb$ is a pseudo-embedding of $T_{\bigcup_{i\leq j} S_i}$ into $N$.
  It remains to show that $\emb(xy)$ begins with $\psi(xy)$ for each $xy\in\bigcup_{i\leq j} S_i$.
  If $xy\in S_j$, then $\emb_j(xy)=\emb(xy)$ begins with $\psi_j(xy)=\psi(xy)$.
  If $xy\in S_i$ for some~$i\ne j$, then $\emb'(xy)=\emb(xy)$ begins with $\psi(xy)$.
\end{proof}

\begin{lemma}\label{lem:DP_correct2}
  Let~$(v,S,\psi,j)$ be a signature with $j>0$,
  let~$q_1,q_2,\dots$ be the children of $v$ in~$\Gamma$,
  let~$S_j:=\{wz\in S \mid L(T_z)\subseteq L(\Gamma_{q_j})$, and
  let~$\psi_j$ be the restriction of $\psi$ to $S_j$.
  Let~$xy\in S$ with $\psi(xy)\in\inA[N]{v}$.

  Then, $\chi[v,S,\psi,j]=\textsc{True}$
  if and only if
  $\chi[v,S,\psi,j-1]=\textsc{True}$ and
  $\outA[T]{y}$ is $v$-splittable into~$(Y_1,Y_2,\ldots)$ and
  \begin{enumerate}[(a)]
    \item $Y_j=\outA[T]{y}$ and there is some~$vw\in\outA[N]{v}$ with $L(T_y)\subseteq L(\Gamma_w)$ and
      $\chi[q_j, S_j, \replace{\psi_j}{xy}{vw}, \outdeg[\Gamma]{q_j}] = \textsc{True}$, or
    \item $\outA[T]{y}\ne\emptyset$ and there is a valid function~$\psi'_j:Y_j\to\outA[N]{v}\cap GW_{q_j}(\Gamma)\setminus\psi_j(S_j)$ such that
      $\chi[q_j,(S_j\setminus\{xy\})\cup Y_j,\psi_j\cup\psi'_j,\outdeg[\Gamma]{q_j}]=\textsc{True}$.
  \end{enumerate}
\end{lemma}
\begin{proof}
  “$\Rightarrow$”:
  Note that all “recursive calls” in the claim correspond to signatures by \cref{lem:sig}, allowing us to apply \cref{lem:last-entry}.
  The witness for $(v,S,\psi,j)$ might correspond to \cref{def:witness}\ref{it:v-split} or \ref{it:S-xy v-split} and we consider the cases individually.

  \textbf{Case 1}:
  $S$ is $v$-splittable into $(S'_1,S'_2,\ldots)$ and
  there is a pseudo-embedding $\emb$ of $T_{\bigcup_{i\leq j} S'_i}$ into $N$ such that,
  for each $wz\in\bigcup_{i\leq j} S_i'$, the path~$\emb(wz)$ begins with the arc~$\psi(wz)$.
  Then, by \cref{obs:subforest-pseudoemb}, $\emb$ is also a pseudo-embedding of $T_{\bigcup_{i\leq j-1} S_i}$,
  so $\emb$ satisfies \cref{def:witness}\ref{it:v-split} for $j-1$, implying $\chi[v,S,\psi,j-1]=\textsc{True}$.
  Since~$S=\bigcup_i S'_i$, we have $xy\in S'_k$ for some~$k$ and, thus, $L(T_y)\subseteq L(\Gamma_{q_k})$.
  For each $yz\in\outA[T]{y}$, we have $L(T_z)\subseteq L(T_y)\subseteq L(\Gamma_{q_k})$,
  so $\outA[T]{y}$ is $v$-splittable into $(Y_1,Y_2,\ldots)$ where $Y_k=\outA[T]{y}$ and $Y_i=\emptyset$ for all $i\ne k$.
  
  If $\outA[T]{y}=\emptyset$, then $\emb(xy)$ cannot end in~$v$ lest it contradict \cref{def:pseudo-embedding}\ref{it:leaves}.
  If $j=k$, then $\emb$ is a pseudo-embedding of $T_{S_j}$.
  In both cases, $\emb(xy)$ begins with $(u, v, w, \ldots)$ for some parent~$u$ of $v$ and child $w$ of $v$ in $N$ with $L(T_y)\subseteq L(\Gamma_w)$.
  Let $\emb'$ result from $\emb$ by clipping~$uv$ off~$\emb(xy)$.
  Since no arc of $T_{S}$ can be mapped to a path with endpoint~$x$ by $\emb$, we know that~$\emb'$ is also a pseudo-embedding.
  Thus, $\chi[q_j,S_j,\replace{\psi_j}{xy}{vw},\outdeg[\Gamma]{q_j}]=\textsc{True}$ by \cref{lem:last-entry}.
  
  If $j\ne k$ and $\outA[T]{y}\ne\emptyset$.
  Then, $Y_j=\emptyset$ so the only valid function~$\psi'_j$ is the empty function.
  But then, $xy\notin S'_j=S_j$ and (b)~degenerates to $\chi[q_j,S_j,\psi_j,\outdeg[\Gamma]{q_j}]=\textsc{True}$
  which holds since, by \cref{obs:subforest-pseudoemb}, $\emb$ is a pseudo-embedding of $T_{S_j}$ into $N$ and
  $\emb(a)$ begins with $\psi_j(a)=\psi(a)$ for all $a\in S_j$.

  \textbf{Case 2}:
  $S\setminus\{xy\}$ is $v$-splittable into $(S'_1,S'_2,\ldots)$ and
  $\outA[T]{y}$ is non-empty and $v$-splittable into $(Y_1,Y_2,\ldots)$ and
  there is a pseudo-embedding~$\emb$ of $T_{\bigcup_{i\leq j} S'_i\cup Y_i}$ into $N$ such that,
  for each $wz\in\bigcup_{i\leq j} S'_i$, the path $\emb(wz)$ begins with $\psi(wz)$ and
  for each $yz\in\bigcup_{i\leq j} Y_i$, the startpoint of $\emb(wz)$ is $v$, implying $\emb(z)\belo{N} v$.
  Let $\psi'_j$ map each $yz\in Y_j$ to the first arc of $\emb(yz)$ and note that $\psi'_j(yz)\in\outA[N]{v}$ for all such $yz$.
  Then, $\psi'_j$ is injective by \cref{def:pseudo-embedding}\ref{it:disjoint paths}.
  Further, $\emb(z)\belo{\Gamma} v$ so $\emb(z)\beloeq{\Gamma} q_j$, implying $L(T_z)\subseteq L(\Gamma_{q_j})$ by \cref{obs:emb-leaves}.
  Thus, $\psi'_j$ is valid and
  $\emb$ is a pseudo-embedding of $T_{(S_j\setminus \{xy\})\cup Y_j}$ such that $\emb(yz)$ begins with $\psi'_j(yz)$ for each $yz\in Y_i$
  and $\emb(wz)$ begins with $\psi_j(wz)$ for each $wz\in S_j$.
  Thus, by \cref{lem:last-entry}, $\chi[q_j,(S_j\setminus\{xy\})\cup Y_j,\psi_j\cup\psi'_j,\outdeg[\Gamma]{q_j}]=\textsc{True}$.

  “$\Leftarrow$”:
  Let~$\chi[v,S,\psi,j-1]=\textsc{True}$ and
  let~$\outA[T]{y}$ be $v$-splittable into $(Y_1,Y_2,\ldots)$.
  We consider the lemma's cases~(a) and (b) individually.

  \textbf{Case~(a)}:
  $Y_j=\outA[T]{y}$ and there is an arc~$vw\in\outA[N]{v}$ such that $L(T_y)\subseteq L(\Gamma_w)$ and
  $\chi[q_j,S_j,\replace{\psi_j}{xy}{vw},\outdeg[\Gamma]{q_j}]=\textsc{True}$.
  By \cref{lem:last-entry},
  there is a pseudo-embedding~$\emb^*$ of $T_{S_j}$ into $N$ such that $\emb^*(st)$ begins with $\replace{\psi_j}{xy}{vw}(st)$ for each $st\in S_j$.
  We claim that the result~$\emb$ of prepending $\psi(xy)$ to $\emb^*(xy)$ is a pseudo-embedding of $T_{S_j}$ into $N$ and,
  for each $st\in S_j$, the path~$\emb(st)$ begins with $\psi(st)$.
  First, no path~$\emb^*(st)$ starts at or above $\psi(xy)$ in $N$, so $\emb$ satisfies \cref{def:pseudo-embedding}\ref{it:disjoint paths}.
  The conditions~\ref{it:end is start} and \ref{it:leaves} of \cref{def:pseudo-embedding} are not impacted by prepending $\psi(xy)$ to $\emb^*(xy)$.
  Further, each path~$\emb(st)$ starts at $\psi_j(st)=\psi(st)$ except $\emb(xy)$ which starts at $\psi(xy)$.
  In the following, we consider the three cases that allow for $\chi[v,S,\psi,j-1]=\textsc{True}$ by \cref{def:witness}.

  First, suppose \cref{def:witness}\ref{it:j0} applies, that is,
  $j=1$ and $S\setminus\{xy\}$ is $v$-splittable (we are using the weaker of the two conditions).
  However, since $Y_j=\outA[T]{y}$, we know that $S$ also admits a $v$-split $(S_1,S_2,\ldots)$
  (note that $S_j$ appears in this $v$-split and
  this $v$-split results from the $v$-split $(S'_1,S'_2,\ldots)$ of $S\setminus\{xy\}$ by adding $Y_j$ to $S'_j$)
  and $xy\in S_j$.
  Thus, $\emb$ is a witness for $(v,S,\psi,j)$ by \cref{def:witness}\ref{it:v-split}.

  Second, suppose \cref{def:witness}\ref{it:v-split} applies, that is,
  $j>1$ and $S$ is $v$-splittable into $(S_1,S_2,\ldots)$ (note that $S_j$ appears in this $v$-split) and
  there is a pseudo-embedding~$\emb'$ of $T_{\bigcup_{i\leq j-1} S_i}$ into $N$ such that, 
  for each $st\in \bigcup_{i\leq j-1} S_i$, the path~$\emb'(st)$ begins with $\psi(st)$.
  But then, $\emb\cup\emb'$ is a pseudo-embedding for $T_{\bigcup_{i\leq j} S_i}$ and,
  for each $st\in \bigcup_{i\leq j} S_i$, the path~$(\emb\cup\emb')(st)$ begins with $\psi(st)$.
  Thus, $\emb\cup\emb'$ is a witness for $(v,S,\psi,j)$ by \cref{def:witness}\ref{it:v-split}.

  Third, suppose \cref{def:witness}\ref{it:S-xy v-split} applies, that is,
  $j>1$ and $S\setminus\{xy\}$ is $v$-splittable into $(S'_1,S'_2,\ldots)$ and
  there is a pseudo-embedding~$\emb'$ of $T_{\bigcup_{i\leq j-1} S'_i\cup Y_i} = T_{\bigcup_{i\leq j-1} S'_i}$ into $N$ such that,
  for all~$st\in \bigcup_{i\leq j-1} S'_i\cup Y_i=\bigcup_{i\leq j-1} S'_i$, the path $\emb'(st)$ begins with $\psi(st)$ and,
  for all~$yz\in\outA[T]{y}$, the startpoint of $\emb'(yz)$ is $v$.
  Now, since $Y=\outA[T]{y}$, we have~$S'_i=S_i$ for all $i \leq j-1$ and $(S_1,S_2,\ldots)$ is a $v$-split of $S$.
  Again, $\emb\cup\emb'$ is a pseudo-embedding for $T_{\bigcup_{i\leq j} S_i}$ and,
  for each $st\in \bigcup_{i\leq j} S_i$, the path~$(\emb\cup\emb')(st)$ begins with $\psi(st)$.
  Thus, $\emb\cup\emb'$ is a witness for $(v,S,\psi,j)$ by \cref{def:witness}\ref{it:v-split}.

  \textsc{Case~(b)}:
  $\outA[T]{y}\ne\emptyset$ and there is a valid function~$\psi'_j:Y_j\to\outA[N]{v}\cap GW_{q_j}(\Gamma)\setminus\psi_j(S_j)$ such that
  $\chi[q_j,(S_j\setminus\{xy\})\cup Y_j,\psi_j\cup\psi'_j,\outdeg[\Gamma]{q_j}]=\textsc{True}$, that is, by \cref{lem:last-entry},
  there is a pseudo-embedding $\emb^*$ of $T_{(S_j\setminus\{xy\})\cup Y_j}$ into $N$ such that
  $\emb^*(st)$ begins with $(\psi_j\cup\psi'_j)(st)$ for each $st\in (S_j\setminus\{xy\})\cup Y_j$.
  Recall that $\outA[T]{y}$ is non-empty and $v$-splittable into~$(Y_1,Y_2,\ldots)$.
  Again, we consider the three cases that allow for $\chi[v,S,\psi,j-1]=\textsc{True}$ by \cref{def:witness}.

  First, suppose \cref{def:witness}\ref{it:j0} applies, that is,
  $j=1$ and $S\setminus\{xy\}$ is $v$-splittable into $(S'_1,S'_2,\ldots)$ (we are using the weaker of the two conditions),
  where $S'_j=S_j\setminus\{xy\}$.
  Then, $\emb^*$ is a pseudo-embedding of $T_{S'_j\cup Y_j}$ into $N$ such that,
  for all~$st\in S'_j$, the path $\emb^*(st)$ begins with $\psi_j(st)=\psi(st)$ and,
  for all~$yz\in Y_j$, the path $\emb^*(yz)$ begins with $\psi'_j(yz)\in \outA[N]{v}$, so its startpoint is~$v$.
  Thus, $\emb^*$ is a witness for $(v,S,\psi,j)$ by \cref{def:witness}\ref{it:S-xy v-split}.
  
  Second, suppose \cref{def:witness}\ref{it:v-split} applies, that is,
  $j>1$ and $S$ is $v$-splittable into $(S_1,S_2,\ldots)$ (note that $S_j$ appears in this $v$-split) and
  there is a pseudo-embedding~$\emb'$ of $T_{\bigcup_{i\leq j-1} S_i}$ into $N$ such that, 
  for each $st\in \bigcup_{i\leq j-1} S_i$, the path~$\emb'(st)$ begins with $\psi(st)$.
  Then, $\emb'$ maps into paths below $\bigcup_{i\leq  j-1} S_i$ in $\Gamma$ and $\emb^*$ maps into paths below $(S_j\setminus\{xy\})\cup Y_j$ in $\Gamma$,
  so no paths mapped to by $\emb'$ and $\emb^*$ can intersect, implying that $\emb'\cup\emb^*$ is a pseudo-embedding.
  Since $\bigcup_i S_i=S$, there is some $k$ with $xy\in S_k$,
  so $L(T_y)\subseteq L(\Gamma_{q_k})$ and $Y_k=\outA[T]{y}$ and $Y_i=\emptyset$ for all $i\ne k$.
  If $k\ne j$, then $\emb^*$ is a pseudo-embedding of $T_{S_j}$ into $N$ such that
  $\emb^*(st)$ begins with $\psi_j(st)=\psi(st)$ for each $st\in S_j$.
  Thus, $\emb:=\emb^*\cup\emb'$ satisfies the conditions of \cref{def:witness}\ref{it:v-split}, so it is a witness for $(v,S,\psi,j)$.
  If $k=j$, then $\emb:=\emb^*\cup\emb'$ is a pseudo-embedding for $T_{\bigcup_{i\leq j}S'_i\cup Y_i}$ into $N$, where $S'_i:=S_i\setminus\{xy\}$ and,
  for each~$st\in \bigcup_{i\leq j} S'_i$, the path~$\emb(st)$ begins with $\psi_j(st)=\psi(st)$ and,
  for each~$yz\in Y_j=\bigcup_{i\leq j} Y_i$, the path~$\emb(yz)$ begins with $\psi'_j(yz)\in\outA[N]{v}$, so its startpoint is~$v$.
  Thus, $\emb:=\emb^*\cup\emb'$ satisfies the conditions of \cref{def:witness}\ref{it:S-xy v-split}, so it is a witness for $(v,S,\psi,j)$.

  Third, suppose \cref{def:witness}\ref{it:S-xy v-split} applies, that is,
  $j>1$ and $S\setminus\{xy\}$ is $v$-splittable into $(S'_1,S'_2,\ldots)$ and
  there is a pseudo-embedding~$\emb'$ of $T_{\bigcup_{i\leq j-1} S'_i\cup Y_i} = T_{\bigcup_{i\leq j-1} S'_i}$ into $N$ such that,
  for all~$st\in \bigcup_{i\leq j-1} S'_i\cup Y_i=\bigcup_{i\leq j-1} S'_i$, the path $\emb'(st)$ begins with $\psi(st)$ and,
  for all~$yz\in\outA[T]{y}$, the startpoint of $\emb'(yz)$ is $v$.
  Again, $\emb'$ maps into paths below $\bigcup_{i\leq j-1} S'_i\cup Y_i$ in $\Gamma$ and $\emb^*$ maps into paths below $S'_j\cup Y_j$ in $\Gamma$,
  so no paths mapped to by $\emb'$ and $\emb^*$ can intersect, implying that $\emb'\cup\emb^*$ is a pseudo-embedding.
  Further, for each~$st\in\bigcup_{i\leq j} S'_i$, the path~$\emb(st)$ begins with $\psi(st)$ and,
  for each~$yz\in\bigcup_{i\leq j} Y_i$, the startpoint of $\emb(yz)$ is $v$ (if $yz\notin Y_j$) or $\emb(yz)$ begins with $\psi'_j(yz)$ (if $yz\in Y_j$),
  but since $\psi'_j(st)\in\outA[N]{v}$, the startpoint is $v$.
  Thus, $\emb$ is a witness for $(v,S,\psi,j)$.
\end{proof}

\paragraph*{Running Time.}
We first bound the number of entries in our lookup table.

\begin{lemma}\label{lem:BSigBound}
  There are $O(2^{b + b \log b}\cdot |A(N)|)$
  signatures~$(v,S,\psi,j)$,
  where $b:=\max_v|GW_v(\Gamma)|$.
\end{lemma}
\begin{proof}
  Since $\psi$ is an injective function, we have $|S| \leq |GW_v(\Gamma)|$ for any signature $(v,S,\psi,j)$,
  \cref{lem:subforestBound} implies that the number of possible values for $S$ is at most $2^{2b}$.
  Given $v$ and $S$, the number of possible injective functions $\psi:S\to GW_v(\Gamma)$ is $b! / ((b-|S|)!) \leq b! \in O(2^{b\log b - b})$.
  Thus, in total, for each pair~$(v,j)$ there are $O(2^{b + b \log b})$ signatures.
  Finally, each such pair corresponds one-to-one to an arc in $\Gamma$ of which there are $|A(N)|$.
\end{proof}

\begin{theorem}\label{thm:TC_FPT}
  Given a phylogenetic network~$N$ and a tree~$T$ on the same leaf-set~$L$,
  as well as a tree extension~$\Gamma$ for~$N$ of width~$k$,
  \Cref{alg:compute-chi} decides whether $N$ displays $T$
  in $O(4^{k + k\log k} n + nm^2)$~time,
  where $n$ and $m$ are the number of nodes and arcs of $N$.
\end{theorem}
\begin{proof}
  Let us abbreviate~$n:=|V(N)|$ and $m:=|A(N)|$.
  By \cref{lem:BSigBound}, the number of signatures~$(v,S,\psi,j)$ is in $O(2^{k+k\log k}\cdot n)$,
  since there are $O(|E(\Gamma)|)=O(n)$ valid pairs~$(v,j)$.
  In order to answer a call with $j=0$, we may need to check if $S$ or $S\setminus \{xy\}$ is $v$-splittable (where $\psi(xy)\in\inA[N]{v}$).
  However, by the nature of the tree extension $\Gamma$, 
  the recursion only produces queries with sets~$S$ that are either $v$-splittable or contain some $xy$ with $\psi(xy)\in\inA[N]{v}$
  such that $S\setminus\{xy\}$ is $v$-splittable.\footnote{Indeed, line~3 in the algorithm can be replaced by “\textbf{return} true;”.}\todo{MW: needs proof?}
  Further, in line~9, we may need to check for some pair $v\in V(N)$ and $y\in V(T)$ whether $\outA[T]{y}$ is $v$-splittable.
  To this end, we precompute a table~$\lambda$ mapping each~$y\in V(T)$ to the lowest node~$v$ in $\Gamma$ with $L(T_y)\subseteq L(\Gamma_v)$,
  that is, $\lambda(y) := LCA_{\Gamma}(L(T_y))$.
  This can be done in $O(n)$ time using range-minimum-query-based LCA preprocessing (see \cite{conf-BF00})
  by computing the LCA in $\Gamma$ of the set of $\lambda(z_i)$ of all children~$z_i$ of $y$ in $T$.
  Indeed, $\outA[T]{y}$ is $v$-splittable if and only if $\lambda(z_i)\belo{\Gamma} \lambda(y) \beloeq{\Gamma} v$ for all $i$.
  In order to compute that actual $v$-split~$(Y_1,Y_2,\ldots)$, we can modify the LCA-routine to also return the first node
  on the path in $\Gamma$ from $\lambda(y)$ to $\lambda(z_i)$ for each child~$z_i$ of $y$ in $T$.
  The grouping of $\outA[T]{y}$ by this first node is the $v$-split we are looking for.
  
  For each signature, we make $\outdeg[N]{v}$ recursive calls if $Y_j=\outA[T]{y}$ and,
  for each~$v$ there is at most one~$j$ with $Y_j=\outA[T]{y}$.
  As $\outdeg[\Gamma]{v}\leq\outdeg[N]{v}$, the total time spent preparing these recursive calls is in $O(m)$.

  Finally, if $\outA[T]{y}\ne\emptyset$,
  then we make another recursive call for each valid $\psi'_j:Y_j \to \outA[N]{v}\cap GW_{q_j}(\Gamma)\setminus\psi_j(S_j)$,
  prepended by a set-union in the second and third argument which can be amortized onto the construction of~$\psi'_j$.\todo{MW: needs proof?}
  However, since $|\outA[N]{v}\cap GW_{q_j}(\Gamma)\setminus\psi_j(S_j)|\leq |GW_{q_j}(\Gamma)|\leq k$,
  there are at most $O(2^{k\log k - k})$ such functions and they can be enumerated in $O(2^{k\log k - k})$~time.\todo{MW: needs proof?}
  
  Using memoization to avoid repeated computations for the same signature, this memoization-augmented version
  answers the query
  \Contained{$\rho_N$,$\{\rho'_N\rho_N\}$,$\{\rho'_T\rho_T\to\rho'_N\rho_N\}$,$\outdeg[N]{\rho_N}$} in time\todo{MW: explain this $\rho'_N$ business}
  \begin{align*}
    O(2^{k + k\log k}\cdot n \cdot (m + 2^{k\log k - k}))
    + O(n)
    \subseteq 
    O(4^{k + k\log k}\cdot n + nm^2) \def\qed{}\qedhere\tag*{\qed}
  \end{align*}
\end{proof}

\section{Hardness Assuming ETH}\label{sec:ETHhardness}

\paragraph{Lower Bounds for Disjoint Paths.}
\citet{LMS18} showed that, unless ETH fails, the following problem cannot be solved in $2^{o(k\log k)}\cdot n^{O(1)}$~time.

\probdef%
{a set $C$ of size-$k$ subsets of $[k]\times[k]$}
{is there a permutation $\pi$ of $[k]$ such that $\{(i,\pi(i))\mid i\in[k]\}$ is a hitting set of $C$?}
{$k\times k$ Hitting-Set}
{$kk$HS}{def:kkHS}

\noindent
We reduce \textsc{$k\times k$ Hitting Set} to \textsc{Tree Containment} on networks with scanwidth in $O(k)$, thus showing that
algorithms for \textsc{Tree Containment} that run in $2^{o(sw\log sw)}n^{O(1)}$ are unlikely to exist.
As an intermediate result, we adapt the reduction given by \citet{LMS18} of \textsc{$k\times k$ Hitting Set} to \textsc{Disjoint Paths}
on DAGs of pathwidth~$O(k)$ such that the DAGs constructed in the reduction have scanwidth~$O(k)$ instead.

\tikzstyle{iface0}=[smallvertex0, fill=gray]
\tikzstyle{iface1}=[smallvertex1, fill=gray]
\tikzstyle{iface2}=[smallvertex2, fill=gray]
\tikzstyle{iface3}=[smallvertex3, fill=gray]

\newcommand{\drawgrid}[2]{
  \foreach \i in {0,1,2,3} \node[smallvertex0] (x\i) at ($(#1)+(\i*#2/4.6+#2/6,0)$) {};
  \foreach \i in {0,1,2,3} {
    \node[smallvertex3] (y\i) at ($(#1)+(#2, \i*#2/5+#2/5)$) {};
    \foreach \j in {0,1,2,3} \node[iface1] (w\i\j) at ($(x\j)+(0,\i*#2/5+#2/5)$) {};
  }
  \begin{pgfonlayer}{background}
    \draw[rounded corners=5] (#1) rectangle ($(#1)+(#2,#2)$);
  \end{pgfonlayer}
}
\newcommand{\drawgridpaths}[2]{
  \foreach \i in {0,1,2,3} \node[smallvertex0] (x\i) at ($(#1)+(\i*#2/4.2+#2/5,0)$) {};
  \foreach \i in {0,1,2,3} {
    \node[smallvertex3] (y\i) at ($(#1)+(#2, \i*#2/5+#2/5)$) {};
    \foreach \j in {0,1,2,3} {
    \node[iface2, tiny] (f\i\j) at ($(x\i) + (-#2/8,\j*#2/5 + #2/4)$) {};
    }
  }
  \begin{pgfonlayer}{background}
    \draw[rounded corners=5] (#1) rectangle ($(#1)+(#2,#2)$);
    \foreach \i in {0,1,2,3} {
      \draw[rounded corners=#2, fill=llightgray] ($(f\i0)+(-#2/16,-#2/12)$) rectangle ($(f\i3)+(#2/16,#2/12)$);
      \coordinate (st) at ($(f\i0)+(#2/8,-#2/32)$);
      \draw[rounded corners=#2, fill=llightgray] ($(st)+(-#2/16,-#2/6)$) rectangle ($(f\i3)+(#2/8,#2/12)$);
    }
  \end{pgfonlayer}
}
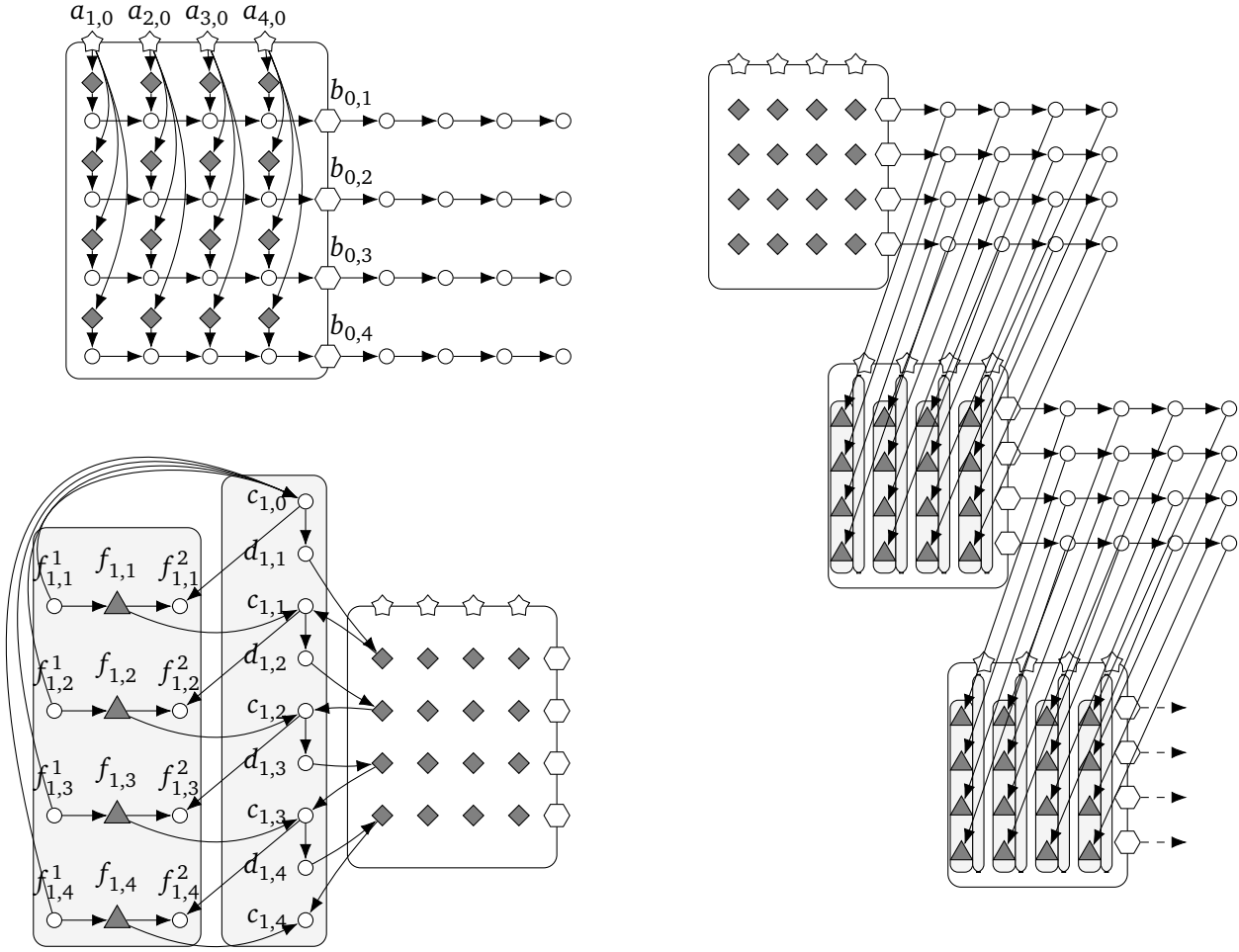
\begin{figure}[t]
  \begin{minipage}[t]{.49\textwidth}\centering
    \begin{tikzpicture}[yscale=-.3, xscale=.7, baseline={(0,2)}, label distance=-3pt]
      \draw[rounded corners=5] (0,-1) rectangle (5,14);

      \foreach [count=\y from 0] \ang in {0,0,0,0}{
      \node[smallvertex] (u0\y) at (.5, 3.5*\y+2.5) {};
        \foreach [count=\x from 1] \st in {,,,3,,,,} {
          \pgfmathtruncatemacro{\prevx}{\x - 1}
          \node[smallvertex\st] (u\x\y) at ($(u\prevx\y)+(1.125,0)$) {} edge[revarc] (u\prevx\y);
        }
        \foreach \x in {0,1,2,3} \node[iface1] (v\x\y) at ($(u\x\y)+(-90+\ang:1.7)$) {} edge[arc] (u\x\y);
        \pgfmathtruncatemacro{\nexty}{\y + 1}
        \node at ($(u4\y)+(-70:1.3)$) {$b_{0,\nexty}$};
      }
      \foreach \x in {0,1,2,3}{
        \pgfmathtruncatemacro{\nextx}{\x + 1}
        \node[smallvertex0, label=above:$a_{\nextx, 0}$] (a\x) at (\x*1.1+.5, -1) {};
        \foreach [count=\y from 0] \b in {0,10,10,10} \draw[arc] (a\x) to[bend right=\b] (v\x\y);
      }
    \end{tikzpicture}\\
    \begin{tikzpicture}[yscale=-.7, xscale=.56,baseline={(0,3)}]
      \useasboundingbox (-6.5,-3) rectangle (5,6.5);
      \drawgrid{0,0}{5}
      \node[smallvertex, label=left:$c_{1,0}$] (c0) at (-1,-2) {};
      \foreach \j in {1,...,4} {
        \pgfmathtruncatemacro{\prevj}{\j-1}
        \node[smallvertex, label=left:$d_{1,\j}$] (d\j) at ($(c\prevj)+(0,1)$) {} edge[revarc] (c\prevj) edge[arc, bend left=5] (w\prevj0);
        \node[smallvertex, label=left:$c_{1,\j}$] (c\j) at ($(d\j)+(0,1)$) {} edge[revarc, bend right=5] (w\prevj0);
        
        \foreach [count=\i from 0] \a/\st in {1/,/iface2,2/} \node[smallvertex, \st, label=90:$f^{\a}_{1,\j}$] (f\i\j) at ($(c\j)+(-6+1.5*\i,0)$) {};
        \draw[arc] (f0\j) -- (f1\j);
        \draw[arc] (f1\j) -- (f2\j);
        \draw[arc] (c\prevj) -- (f2\j);
        \draw[arc] (f1\j) to[bend left=20] (c\j);
        \draw[arc] (f0\j) .. controls ($(-8,-2)+(-135:\j)$) and ($(c0)+(-160:3)$) .. (c0);
      }
      \begin{pgfonlayer}{background}
        \draw[rounded corners=5, fill=llightgray] ($(c0)+(-2,-.5)$) rectangle ($(c4)+(.5,.5)$);
        \draw[rounded corners=5, fill=llightgray] ($(f01)+(-.5,-1.5)$) rectangle ($(f24)+(.5,.5)$);
      \end{pgfonlayer}
    \end{tikzpicture}
  \end{minipage}
  \hspace{0mm}
  \begin{minipage}[t]{.49\textwidth}\centering
    \begin{tikzpicture}[yscale=-1, xscale=.8, baseline={(0,2)}]
      \drawgrid{0,0}{3}
      \foreach \i in {0,1,2,3} {
        \node[smallvertex] (z\i0) at ($(y\i)+(1,0)$) {} edge[revarc] (y\i);
        \foreach \j in {1,2,3} {
          \pgfmathtruncatemacro{\prevj}{\j-1}
          \node[smallvertex] (z\i\j) at ($(z\i\prevj)+(.9,0)$) {} edge[revarc] (z\i\prevj);
        }
      }
      \drawgridpaths{2,4}{3}
      \foreach \i in {0,1,2,3} {
        \foreach \j in {0,1,2,3} {
          \draw[arc] (z\i\j) to (f\j\i);
        }
      }
      \foreach \i in {0,1,2,3} {
        \node[smallvertex] (z\i0) at ($(y\i)+(1,0)$) {} edge[revarc] (y\i);
        \foreach \j in {1,2,3} {
          \pgfmathtruncatemacro{\prevj}{\j-1}
          \node[smallvertex] (z\i\j) at ($(z\i\prevj)+(.9,0)$) {} edge[revarc] (z\i\prevj);
        }
      }
      
      \drawgridpaths{4,8}{3}
      \foreach \i in {0,1,2,3} {
        \foreach \j in {0,1,2,3} {
          \draw[arc] (z\i\j) to (f\j\i);
        }
      }
      \foreach \i in {0,1,2,3} \draw[dashed, arc] (y\i) --++(1,0);

    \end{tikzpicture}
  \end{minipage}

  \caption[Example of \cref{cons:kkHS to DP}]{Example of \cref{cons:kkHS to DP}. Top left: gadget~$Q_\ell$. Bottom left: gadget~$R_\ell$. Right: the~$Q_i$ and the~$R_i$ are strung together such as to form a chain whose cutwidth is bounded in $k$.}
  \label{fig:kkHS to TC}
\end{figure}

The following construction is a slight modification of the construction of \citet{LMS18} showing that,
unless ETH fails, \textsc{Disjoint Paths} cannot be solved in $O^*(2^{o(\omega \log \omega)})$~time,
where $\omega$ is the pathwidth of the undirected graph underlying the input DAG.
We try to keep our notation as close to the original reduction as possible.

\begin{construction}[see \cref{fig:kkHS to TC}]\label{cons:kkHS to DP}
  Let $\mathcal{C}=\{S_1,S_2,\ldots\}$ be an instance of \textsc{$k\times k$ Hitting Set}.
  For each set $S_\ell$, we create two gadgets~$Q_\ell$ and $R_\ell$.
  Then, we string these gadgets together by adding arcs from $Q_\ell$ to $R_\ell$ and identifying some vertices in $R_\ell$ with vertices in $Q_{\ell+1}$.

  \medskip\noindent
  First, we create a gadget $Q_\ell$ as follows (see \cref{fig:kkHS to TC} (top left)).
  \begin{enumerate}[(Q1)]
    \item for each $1\leq i\leq k$, create a vertex $a_{i,0}$
      \todo{MW: remove the 0 index for better readability?}
    \item for each $1\leq j\leq k$, create a path $(a_{1,j},a_{2,j},\ldots,a_{k,j}, b_{0,j},b_{1,j},b_{2,j},\ldots,b_{k,j})$
    \item for each $1\leq i,j\leq k$, create a vertex $v_{i,j}$ with arcs $a_{i,0}v_{i,j}$ and $v_{i,j}a_{i,j}$
  \end{enumerate}
  Note that we consistently represent
  vertices $a_{i,0}$  $b_{1,j}$, and $v_{i,j}$ by
  \tikzvert{smallvertex0}, \tikzvert{smallvertex3}, and \tikzvert{iface1}, respectively.

  \medskip\noindent
  Next, we construct a second gadget $R_\ell$ as follows (see \cref{fig:kkHS to TC} (bottom left)).
  \begin{enumerate}[(R1)]
    \item for each $1\leq i\leq k$, create a path $(c_{i,0},d_{i,1},w_{i,1},c_{i,1},d_{i,2},w_{i,2},\ldots,c_{i,k})$.
      Each vertex $w_{i,j}$ will later be identified with $v_{i,j}$ of the gadget $Q_{\ell+1}$ (unless $\ell=|\mathcal{C}|$).
    \item for each $1\leq i,j\leq k$, create the path $(f^1_{i,j}, f_{i,j}, f^2_{i,j})$ and arcs
      $f^1_{i,j}c_{i,0}$, $c_{i,j-1}f^2_{i,j}$ and $f_{i,j}c_{i,j}$.
    \item add vertices $s$, $t$ and, for each $(i,j)\in S_\ell$, add the arcs $sd_{i,j}$ and $d_{i,j}t$ (not drawn in \cref{fig:kkHS to TC}).
  \end{enumerate}
  Note that we consistently represent vertices $f_{i,j}$ by \tikzvert{iface2}.

  \medskip\noindent
  To connect $Q_\ell$ to $R_\ell$ add the arc $b_{i,j}f_{i,j}$ for all $1\leq i,j\leq k$ (see \cref{fig:kkHS to TC} (right)).

  \medskip\noindent
  As the demands for disjoint paths, we add:
  \begin{enumerate}[(D1)]
    \item from $a_{i,0}$ to $c_{i,k}$ for each $1\leq i\leq k$,
    \item from $f^1_{i,j}$ to $f^2_{i,j}$ for each $1\leq i,j\leq k$,
    \item from $s$ to $t$.
  \end{enumerate}
  Note that all demands are from vertices of indegree zero.

  \medskip\noindent
  Finally, for each $1\leq \ell < |\mathcal{C}|$ and
  for all $1\leq i,j\leq k$, identify the vertex $w_{i,j}$ in gadget $R_\ell$ with the vertex $v_{i,j}$ in $Q_{\ell+1}$
  (see \cref{fig:kkHS to TC} (right)).
\end{construction}

\noindent
Let $\mathcal{H}$ be an instance of $k\times k$-\textsc{Hitting Set} and
let $(G,D)$ be the result of applying \cref{cons:kkHS to DP} to $\mathcal{H}$,
where $G$ is a DAG and $D$ is a set of demands for disjoint paths.
Next, consider the result~$G^*$ of contracting, in each~$Q_\ell$, each path~$(a_{1,j},a_{2,j},\ldots,b_{k,j})$ to a single vertex~$b_j$
for each $j$ and observe, that $(G^*,D)$ is exactly the result of applying the reduction of \citet{LMS18} to $\mathcal{H}$.
Therefore, to show correctness of our modified reduction, the following lemma suffices.

\begin{lemma}\label{lem:disjoint_paths_contract}
  Let $v$ be any vertex in $G$ with in- or out-degree one that does not occur in $D$ and
  let $G'$ be the result of contracting $v$ onto its unique parent or child, respectively.
  Then, $D$ can be satisfied in $G$ if and only if $D$ can be satisfied in $G'$.
\end{lemma}
\begin{proof}
  Suppose that $v$ has in-degree one as the proof is analogous for the case that $v$ has out-degree one.
  Let $u$ be the unique parent of $v$.

  ``$\Rightarrow$'':
  Let $P$ be a set of paths in $G$ satisfying $D$.
  Suppose there is a path $p\in P$ containing $v$ since, otherwise, this direction clearly holds.
  Since $v$ does not occur in $D$, we know that $p$ contains the unique incoming arc $uv$ of $v$.
  But then, replacing $p$ by the result $p'$ of contracting $uv$ in $p$ results in a path-set $P'$ satisfying $D$ in $G'$.

  ``$\Leftarrow$'':
  Similarly to ``$\Rightarrow$'', we modify the path $p'$ containing the vertex~$w$ in $G'$ that corresponds to the result of the contraction of $v$
  with the subpath $(u,v)$.
\end{proof}

\noindent\looseness=-1
Now, the whole reason why we modified the construction of \citet{LMS18} is to lower the scanwidth of the construction to $O(k)$.
To prove this bound, we define \emph{blocks} of vertices that we will string up to form a (reverse) topological order with the desired scanwidth.
We give topological orders $\sigma_\ell$ for parts of $G$ such that $(\sigma_0,\sigma_1,\sigma_2,\ldots,\sigma_{|\mathcal{C}|})$ is a topological order and
$G$ has cutwidth (and, thus, scanwidth) in $O(k)$.
The high-level idea is to cut each gadget in vertical slices, taking care that all vertices of a slice only have out-neighbors in the same or the next slice.

First, $\sigma_0$ is the concatenation of the orders $(a_{i,0},v_{i,1}, v_{i,2},\ldots,v_{i,k},a_{i,1},a_{i,2},\ldots,a_{i,k})$
for $1\leq i \leq k$, thus covering the vertices of $Q_1$.
Second, $\sigma_\ell$ with $\ell\geq 1$ starts with the $s$-vertex of $R_\ell$ and ends with the $t$-vertex of $R_\ell$ and
is otherwise made up of blocks $(B_0,B_1,B_2,\ldots,B_k)$ containing
vertices of $Q_\ell$ (namely the $b_{i,j}$), vertices of $R_\ell$, and vertices of $Q_{\ell+1}$ (if $\ell<|\mathcal{C}|$).
First, $B_0$ contains all $b_{0,j}$ of $Q_\ell$ for $1\leq j\leq k$.
Then, for each $1\leq i\leq k$, the block $B_i$ is
$(X_i^1, \ldots, X_i^k, c_{i,0}, a_{i,0}, Y_i^1,\ldots, Y_i^k)$,
where $a_{i,0}$ is a node from $Q_{\ell+1}$ occuring only if $\ell<|\mathcal{C}|$ and,
for each $1\leq j\leq k$, the sub-orders $X^j_i$ and $Y^j_i$ are defined as
$X_i^j := (b_{i,j}, f^1_{i,j}, f_{i,j})$ where $b_{i,j}$ is from $Q_\ell$, and
$Y_i^j := (f^2_{i,j}, d_{i,j}, w_{i,j}, a_{i,j}, c_{i,j})$, where the $a_{i,j}$ are from $Q_{\ell+1}$
and the $w_{i,j}$ are the nodes $v_{i,j}$ from $Q_{\ell+1}$, both of which occur in $B_i$ only if $\ell<|\mathcal{C}|$.

\newcommand{\srev}{\ensuremath{\sigma^\text{r}}}
\begin{lemma}\label{lem:kkHS sw}
  The cutwidth of $\sigma$ (and, thus, of $G$) is in $O(k)$.
\end{lemma}
\begin{proof}
  Let us count the number of arcs that can possibly be cut by any cut inside a fixed $\sigma_\ell$.
  
  \begin{itemize}
    \item Arcs incoming to $\sigma_\ell$ can only come from $\sigma_{\ell-1}$ and
  arcs outgoing of $\sigma_\ell$ can only go to $\sigma_{\ell+1}$ and there are exactly~$k$ of each of them,
  namely the arcs $a_{k,j}b_{0,j}$.
    \item There are at most $2k$ arcs incident with $s$ or $t$ since each $S_\ell$ has size at most $k$.
    \item Each $B_i$ contains $O(k)$ vertices of constant degree and one vertex (namely $c_{i,0}$) of in-degree~$k$ (and out-degree one),
      so $O(k)$ arcs run within each block~$B_i$.
    \item Arcs with their tail in some $B_i$ with $1\leq i< k$ whose head is not in $B_i$ have their head in $B_{i+1}$.
      They run from $b_{i,j}$ to $b_{i+1,j}$ in $Q_\ell$ and from $a_{i,j}$ to $a_{i+1,j}$ in $Q_{\ell+1}$ and
      there are exactly~$2k$ of them for each $1\leq i<k$
  \end{itemize}
  In total, all cuts in $\sigma$ cut at most $O(k)$ arcs.
\end{proof}

\begin{theorem}\label{thm:sw disjoint paths}
  Unless ETH fails,
  \textsc{Disjoint Paths} in DAGs cannot be solved in $2^{o(k\log k)}n^{O(1)}$~time,
  where $k$ is the cutwidth of the input DAG.
\end{theorem}

\paragraph{From Disjoint Paths to Tree Containment.}
\cref{cons:kkHS to DP} can be modified to show that \textsc{Tree Containment} cannot be solved in $2^{o(k\log k)}n^{O(1)}$~time,
where $k$ is the scanwidth of the input network.
To this end, we turn the input DAG into a (binary) phylogenetic network~$N$ and
we turn the demand-set for disjoint paths into a (binary) phylogenetic tree~$T$.
This reduction borrows ideas of the original NP-hardness reduction for \textsc{Tree Containment} by \citet{KNL+08}.

\begin{figure}[t]
  \centering
  \begin{tikzpicture}
    \node at (-.5,0) {$\sigma:$};
    \foreach \i in {1,...,30} \node[smallvertex] (v\i) at (.4*\i,0) {};

    \node[smallvertex] (r1) at (3,3) {} edge[arc] (v1);
    \nextnode[leaf, label=below:$\ell_1$]{l1}{v1}{-120:.7}{revarc}
    \nextnode[leaf, label=below:$\ell'_1$]{l1}{v1}{-60:.7}{revarc}
    \node[smallvertex, fill=gray] at (v1) {};

    \foreach[count=\i from 2] \tar in {5,9,13,18,22} {
      \pgfmathtruncatemacro{\prev}{\i-1}
      \nextnode{r\i}{r\prev}{-15:1.3}{revarc}
      \draw[arc] (r\i) -- (v\tar);
      \node[smallvertex, fill=gray] at (v\tar) {};
      \nextnode[leaf, label=below:$\ell_{\i}$]{l\i}{v\tar}{-120:.7}{revarc}
      \nextnode[leaf, label=below:$\ell'_{\i}$]{l\i}{v\tar}{-60:.7}{revarc}
    }

    \nextnode[leaf, label=below:$\ell_7$]{l7}{v27}{-120:.7}{revarc}
    \nextnode[leaf, label=below:$\ell'_7$]{l7}{v27}{-60:.7}{revarc}
    \node[smallvertex, fill=gray] at (v27) {};
    \draw[arc] (r6) -- (v27);

  \end{tikzpicture}
  \caption{Illustration of the tree $T$ constructed in \cref{cons:DP to TC}.
  Source-nodes in $G$ are indicated in gray.}
  \label{fig:DP2TC}
\end{figure}
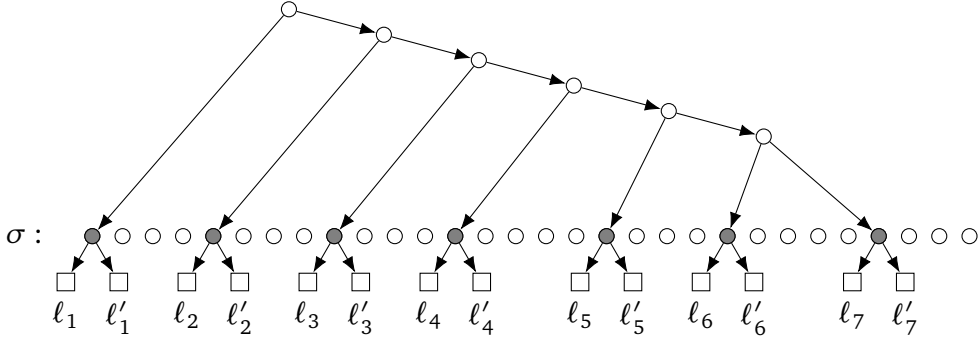

\begin{construction}\label{cons:DP to TC}
  Let $(G,D)$ be an instance of \textsc{Disjoint Paths} created by \cref{cons:kkHS to DP}

  \begin{enumerate}
    \item 
      For each demand~$(x_i,y_i)\in D$ with~$1\leq i\leq |D|$,
      modify~$G$ by adding
      a new leaf labelled~$\ell_i$ to~$x_i$ and 
      a new leaf labelled~$\ell'_i$ to~$y_i$.
      Let~$\sigma'$ result from~$\sigma$ by inserting the new leaves directly behind their respective parent.

    \item
      Add a new node~$\rho$ and add an arc from $\rho$ to every source-node of~$G$.

    \item\label{it:step_contract}
      For every node~$q$ with \emph{out-degree} at least three,
      replace~$q$ with a binary caterpillar whose leaves are the children of~$q$ in order of appearance in~$\sigma'$
      (see \Cref{fig:DP2TC}).
      Likewise, replace any node~$q$ with high \emph{in-degree}~$d$ by a chain of $d-1$~binary reticulations
      whose parents are the parents of~$q$ in order of appearance in~$\sigma$.
      Let~$N$ be the resulting DAG and note that~$N$ is rooted.

    \item
      Construct a tree~$T$ by deleting from~$N$ all nodes that are strict descendants of source-nodes in~$G$ and,
      for every demand~$(x_i,y_i)\in D$, adding leaves labelled~$\ell_i$ and~$\ell'_i$ as children to~$x_i$.
  \end{enumerate}
\end{construction}

\noindent
Note that the cutwidth of $N$ is not much larger than the cutwidth of $G$.
First, observe that adding~$\rho$ and immediately turning it into a caterpillar increases the cutwidth of $\sigma$
by one, since the caterpillar “follows” the linear layout of $\sigma$ (see \Cref{fig:DP2TC}).
Second, adding additional leaves increases the cutwidth by at most one.
Finally, replacing high-degree nodes with caterpillars or reticulation stacks does not increase the cutwidth.

\begin{observation}\label{obs:N_cutwidth}
  Let~$(N,T)$ result from applying \Cref{cons:DP to TC} to $(G,D)$ and
  let~$G$ have cutwidth~$k$.
  Then, $N$ has cutwidth at most~$k+2$.
\end{observation}

The argument that~$N$ displays~$T$ if and only if the demands~$D$ can be satisfied in~$N$ (not in~$G$) follows
by the same arguments as \citet{KNL+08} gave for the original NP-hardness proof.
It remains to show that~$D$ can be satisfied in~$N$ if and only if it can be satisfied in~$G$.

\begin{lemma}\label{lem:D_in_N_and_G}
  Let~$(N,T)$ result from applying \Cref{cons:DP to TC} to $(G,D)$.
  Then,
  $(N,T)$ is a yes-instance of \textsc{Disjoint Paths} if and only if $(G,D)$ is.
\end{lemma}
\begin{proof}
  By \Cref{lem:disjoint_paths_contract},
  we can contract caterpillars and reticulation-chains created in Step~\ref{it:step_contract} back to single nodes
  without changing the satisfiability of~$D$.
  Furthermore, neither~$\rho$ nor any of the newly added leaves can be used by any solution satisfying~$D$ in $N$,
  so they can be removed without changing the satisfiability of~$D$.
  Since applying these operations to $N$ yields~$G$, the claim follows.
\end{proof}

Finally, the main theorem of this section follows from the previous lemmas.

\begin{theorem}\label{thm:sw TC}
  Unless ETH fails,
  \textsc{Tree Containment} cannot be solved in $2^{o(k\log k)}n^{O(1)}$~time,
  where $k$ is the cutwidth of the input network, even if both the network and the tree are binary.
\end{theorem}

\section{Discussion}

In this work, we considered the \textsc{Tree Containment} problem,
a fundamental computational problem arising in the study of evolutionary networks.
While the problem is NP-hard, we presented an algorithm solving it in $O(4^{k + k\log k}n + nm^2)$~time,
where $n$ and $m$ denote the numbers of vertices and arcs in the input network and $k$ is its \emph{scanwidth}.
We further showed that this running time is asymptotically optimal under the Exponential Time Hypothesis.
Our result generalizes known tractability results for bounded-level networks~\cite{vaniersel2010locating}.
%
%
In contrast to a previous treewidth-based algorithm~\cite{vIJW23}, the dynamic programming algorithm presented here exploits the sequential structure induced by scanwidth 
and is comparatively simple and implementation-friendly.

A crucial feature of our algorithm is that it relies on being given a suitable ``tree extension'' of the input network,
which raises several 
algorithmic questions.
While highly optimized practical algorithms and heuristics exist for computing tree decompositions of small width,
the analogous problem for tree extensions has so far only been studied in~\cite{holtgrefe2026exact}, which presented several parameterized exact and heuristic algorithms.
%
%
While it is conjectured that, like treewidth, scanwidth cannot be constant-factor approximated in polynomial time,
it would be interesting to obtain such approximations in single-exponential time, analogous to recent advances for treewidth~\cite{KL23}. In addition, an FPT algorithm 
for finding an optimal tree-extension, parameterized by scanwidth, could be combined with our algorithm to give an FPT algorithm for \textsc{Tree Containment}, parameterized by scanwidth.

Another promising direction is to incorporate notions of invisibility into scanwidth,
following similar results for a “invisible” analog of the level used by \citet{Wel17}.
Such a refinement could potentially capture structural properties of phylogenetic networks
that remain hidden from purely local width measures and might lead to improved algorithms for containment-type problems.

It is natural to investigate whether our techniques extend to \emph{node scanwidth},
the node-based variant of scanwidth in which only the tails of active arcs contribute to the width.
Since node scanwidth can be strictly smaller than scanwidth in networks with high out-degree,
such a generalization could yield improved parameter dependencies for a broader range of networks,
in particular for the version allowing soft polytomies.

Beyond \textsc{Tree Containment}, it would be interesting to determine
which other computational problems on phylogenetic networks admit efficient scanwidth-based algorithms.
Promising candidates include
counting variants of containment problems,
agreement and compatibility problems between networks and trees, and
optimization problems arising in likelihood- or parsimony-based inference.
More generally, the growing evidence that scanwidth captures algorithmically relevant structure that escapes classical undirected width measures
suggests that it may serve as the foundation for a broader theory of directed width measures in phylogenetics.


\printbibliography[notcategory=ignore]

\end{document}